\documentclass[twoside]{article}

\usepackage[accepted]{aistats2020}
\usepackage[round]{natbib}

\usepackage[utf8]{inputenc} % allow utf-8 input
\usepackage[T1]{fontenc}    % use 8-bit T1 fonts
\usepackage[hyphens]{url}            % simple URL typesetting
\usepackage{hyperref}     % hyperlinks (removed by Rudy for now - split links causing compilation problems)
\usepackage{breakurl}
\usepackage{booktabs}       % professional-quality tables
\usepackage{amsfonts}       % blackboard math symbols
\usepackage{nicefrac}       % compact symbols for 1/2, etc.
\usepackage{microtype}      % microtypography

\usepackage{multirow}  % added by sungjin
\usepackage{wrapfig} % added by sungjin

\usepackage{graphicx}

\usepackage[parfill]{parskip}

\usepackage{array}
\usepackage{paralist}
\usepackage{verbatim}
\usepackage{subfig}

\usepackage{algorithmic,algorithm}

\usepackage{amssymb}
\usepackage{enumerate}
\usepackage{amsmath}
\usepackage{amsthm}
\usepackage{bm}
\usepackage{tikz}
\newtheorem{theorem}{Theorem}
\newtheorem{claim}{Claim}
\newtheorem{corollary}{Corollary}
\newtheorem{definition}{Definition}
\newtheorem{lemma}{Lemma}

\newcommand{\defeq}{:=}
\newcommand{\algname}{\textsc{NK-means}}
\newcommand{\framename}{\textsc{SampleCoreset}}
\newcommand{\km}{\textsc{k-means++}}

\begin{document}



\runningauthor{Sungjin Im, Mahshid Montazer Qaem, Benjamin Moseley, Xiaorui Sun, Rudy Zhou}

\twocolumn[

\aistatstitle{Fast Noise Removal for $k$-Means Clustering}

\aistatsauthor{Sungjin Im \And Mahshid Montazer Qaem \And Benjamin Moseley}

\aistatsaddress{University of California at Merced \And University of California at Merced \And Carnegie Mellon University}

\aistatsauthor{Xiaorui Sun \And Rudy Zhou}

\aistatsaddress{University of Illinois at Chicago \And Carnegie Mellon University}

]
\begin{abstract}
    This paper considers $k$-means clustering in the presence of noise. It is known that $k$-means clustering is highly sensitive to noise, and thus noise should be removed to obtain a quality solution. A popular formulation of this problem is called \emph{$k$-means clustering with outliers}. The goal of $k$-means clustering with outliers is to discard up to a specified number $z$ of points as noise/outliers and then find a $k$-means solution on the remaining data. The problem has received significant attention, yet current algorithms with theoretical guarantees suffer from either high running time or inherent loss in the solution quality. The main contribution of this paper is two-fold. Firstly, we develop a simple greedy algorithm that has \emph{provably} strong worst case guarantees. The greedy algorithm  adds a simple preprocessing step to remove noise, which can be combined with any $k$-means clustering algorithm. This algorithm gives the first pseudo-approximation-preserving reduction from $k$-means with outliers to $k$-means without outliers. Secondly, we show how to construct a coreset of size $O(k \log n)$.  When combined with our greedy algorithm, we obtain a scalable, near linear time algorithm. The theoretical contributions are verified experimentally by demonstrating that the algorithm quickly removes noise and obtains a high-quality clustering. 
\end{abstract}

\section{Introduction}

Clustering is a fundamental unsupervised learning method that offers a compact view of  data sets by grouping similar input points. Among various clustering methods,  $k$-means clustering is one of the most popular clustering methods used in practice, which is defined as follows: given a set $X$ of $n$ points in Euclidean space\footnote{The input space can be extended to an arbitrary metric space.} $\mathbb{R}^d$  and a target number of clusters $k$, the goal is to choose a set $C$ of $k$ points from $\mathbb{R}^d$ as centers, so as to minimize the $\ell_2$-loss, i.e., the sum of the squared distances of every point $x \in X$ to its closest center in $C$. 

Due to its popularity, $k$-means clustering has been extensively studied for decades both theoretically and empirically, and as a result, various novel algorithms and powerful underlying theories have been developed. In particular, because the clustering problem is NP-hard, several constant-factor approximation algorithms have been developed \citep{charikar1999improved,KanungoMNPSW04,Kumar2004,feldman2007ptas}, meaning that their output is always within an $O(1)$ factor of the optimum. One of the most successful algorithms used in practice is $k$-means++ \citep{ArthurV07}. The algorithm $k$-means++ is a preprocessing step used to set the initial centers when using Lloyd's algorithm \citep{Lloyd82}.  Lloyd's algorithm is a simple local search heuristic that alternates between updating the center of every cluster and reassigning points to their closest centers. 
%is a popular local search heuristic. 
$k$-means++ has a provable approximation guarantee of $O(\log k)$  by carefully choosing the initial centers.

$k$-means clustering is highly sensitive to noise, which is present in many data sets. Indeed, it is not difficult to see that the $k$-means clustering objective can vary significantly even with the addition of a single point that is far away from the true clusters. In general, it is a non-trivial task to filter out noise; without knowing the true clusters, we cannot identify noise, and vice versa. While there are other clustering methods, such as density-based clustering \citep{eksx96}, that attempt to remove noise, they do not replace $k$-means clustering because they are fundamentally different than $k$-means.
%\todo{I wouldn't say lloyd's is a local search. It's really an alternating method. BENs Reply:  I disagree.  It totally is local search or as we say is OR "improvement hueristic"}

Consequently, there have been attempts to study $k$-means clustering in the presence of noise. The following problem formulation is the most popular formulation in the theory \citep{Chen08,CharikarKMN01,McCutchenK08,guha2017distributed}, machine learning \citep{MalkomesKCWM15,ChawlaG13,li2018distributed} and database communities \citep{GuptaKLMV17}. Note that traditional $k$-means clustering is a special case of this problem when $z = 0$. Throughout, for $x, y \in \mathbb{R}^d$, we let $d(x, y)$ denote the $\ell_2$ distance between $x$ and $y$. For a subset of points $Y$, let $d(x, Y) \defeq \min_{y \in Y} d(x, y)$.

%\smallskip
%\begin{definition}[$k$-Means with Outliers]
%	In the $k$-means with outliers problem, we are given as input a subset $X = \{x_1, \dots, x_n\} \subset M$, a parameter $k \in \mathbb{N}$ (number of centers), and a parameter $z \in \mathbb{N}$ (number of outliers.)
%	The goal is to choose a collection of $k$ centers, $C = \{c_1, \dots c_k\} \subset X$, to minimize:
%	\[\sum\limits_{x \in X_z(C)} d^2(x,C)\]
%	, where $X_z(C) \subset X$ is the subset of $n-z$ input points with the smallest squared distances to $C$.
%\end{definition}

\smallskip
\begin{definition}[$k$-Means with Outliers]
	In this problem we are given as input a subset $X$ of $n$ points in $\mathbb{R}^d$, a parameter $k \in \mathbb{N}$ (number of centers), and a parameter $z \in \mathbb{N}$ (number of outliers).
	The goal is to choose a collection of $k$ centers, $C \subseteq \mathbb{R}^d$, to minimize:
	$\sum_{x \in X_z(C)} d^2(x,C)$
	, where $X_z(C) \subseteq X$ is the subset of $n-z$ input points with the smallest distances to $C$.
\end{definition}

Because this problem generalizes $k$-means clustering, it is NP-hard, and in fact, turns out to be significantly more challenging. The only known constant approximations \citep{Chen08,KrishnaswamyLS18} are highly sophisticated and are based on complicated local search 
or linear program rounding. They are unlikely to be implemented in practice due to their runtime and complexity. Therefore, there have been strong efforts to develop simpler algorithms that offer good approximation guarantees when allowed to discard more than $z$ points as outliers \citep{CharikarKMN01,meyerson2004k,GuptaKLMV17}, or heuristics \citep{ChawlaG13}. Unfortunately, the aforementioned algorithms with theoretical guarantees suffer from either high running time or inherent loss in solution quality.

In contrast , concurrently with our work, \cite{DBLP:conf/nips/BhaskaraVX19} developed a simple algorithm for $k$-means with outliers that gives a good approximation guarantee when allowed to discard more than $z$ outliers or use more than $k$ centers, which we discuss in more detail when we describe our results.

%However, all existing algorithms have running time that is super-quadratic in the input size, or have no guarantees on the approximation ratio and running time. 

%all such existing algorithms significantly deviate from the elegant alternating nature of Lloyd's algorithm, and therefore, has been of limited use: 

%\newcommand{\outa}{NK-\mathsc{Means}~}

\subsection{Our Results and Contributions}

The algorithmic contribution of this paper is two-fold, and further these contributions are validated by experiments. In this section, we state our contribution and discuss it in detail compared to the previous work.

\paragraph{Simple Preprocessing Step for Removing Outliers with Provable Guarantees:} In this paper we develop a simple preprocessing step, which we term \algname,  to effectively filter out outliers. \algname~ stands for noise removal for $k$-means. Our proposed preprocessing step can be combined with \emph{any} algorithm for $k$-means clustering.  Despite the large amount of work on this problem, we give the \emph{first} reduction to the standard $k$-means problem.  In particular, \algname~can be combined with the popular $k$-means++. The algorithm is the fastest known algorithm for the $k$-means with outliers problem.  Its speed and simplicity gives it the potential to be  used in practice. Formally, given an $\alpha$-approximation for $k$-means clustering, we give an algorithm for $k$-means with outliers that is guaranteed to discard up to $O(kz)$ points such that the cost of remaining points is at most $O(\alpha)$ times the optimum that discards up to exactly $z$ points. While the theoretical guarantee on the number of outliers is larger than $z$ on  worst-case inputs, we show that \algname~removes at most $O(z)$ outliers under the assumption that every cluster in an optimal solution has at least $3z$ points. We believe that this assumption captures most practical cases since otherwise significant portions of the true clusters can be discarded as outliers. In actual implementation, we can guarantee discarding exactly $z$ points by discarding the farthest $z$ points from the centers we have chosen. It is worth keeping in mind that all (practical) algorithms for the problem discard more than $z$ points to have theoretical guarantees \citep{CharikarKMN01,meyerson2004k,GuptaKLMV17, DBLP:conf/nips/BhaskaraVX19}.

When compared to the concurrent work of \cite{DBLP:conf/nips/BhaskaraVX19}, our work differs in two main ways. Firstly, our algorithm consists of a pre-processing step that can be combined with any $k$-means algorithm, such as $k$-means++, while theirs consists of a modification to the standard  $k$-means++ algorithm. Further, our algorithm throws away a multiplicative $O(k)$-factor extra outliers in the worst case and uses no extra centers, while in their work they prove a trade-off between extra outliers and extra centers. In particular, they obtain a $O(\log k)$ approximation for $k$-means with outliers using an extra multiplicative $O(\log k)$ outliers and no extra centers, and a $O(1)$-approximation using an extra multiplicative $O(1)$ outliers and centers.
%it is very unlikely to discard a considerably more number of outliers is extremely unlikely to happen unless a cluster size is close to the target  number of outliers -- in such cases, it makes more sense to increase $k$ (choose more centers) or decrease $z$ (discard less outliers). 

\paragraph{New Coreset Construction:} When the data set is large, a dominant way to speed up clustering is to first construct a coreset and then use the clustering result of the coreset as a solution to the original input. Informally, a set of (weighted) points $Y$ is called a coreset of $X$ if a good clustering of $Y$ is also a good clustering of $X$ (see Section \ref{sec:coreset-def} for the formal definition of coreset.)

The idea is that if we can efficiently construct such $Y$, which is significantly smaller than $X$, then we can speed up any clustering algorithm with little loss of accuracy. In this paper, we give an algorithm to construct a coreset of size $O(k \log n)$ for $k$-means with outliers. Importantly, the coreset size is independent of $z$ and $d$ - the number of outliers and dimension, respectively.

%Previously, the best known coreset size was either $O(k + z)$ \citep{GuptaKLMV17} or $O(\frac{kn}{z} \log n)$ \citep{meyerson2004k} -- thus, the coreset size was $\Omega(\sqrt n)$ in general. We achieve a coreset of size $O(k\log n)$ by combining these two coreset construction methods.

\paragraph{Experimental Validation:} 
Our new coreset enables the implementation and comparison of all potentially practical algorithms, which are based on primal-dual  \citep{CharikarKMN01}, uniform sampling \citep{meyerson2004k}, or local search \citep{ChawlaG13,GuptaKLMV17}. It is worth noting that, to the best of our knowledge, this is the first paper to implement the primal-dual based algorithm \citep{CharikarKMN01} and test it for large data sets. 
We also implemented natural extensions of $k$-means++ and our algorithm \algname. We note that for fair comparison, once each  algorithm chose the $k$ centers, we considered all points and discarded the farthest $z$ points. Our experiments show that our $\algname$ consistently outperforms other algorithms for both synthetic and real-world data sets with little running time overhead as compared to $k$-means++. 

\iffalse
Thanks to our new coreset construction, we were able to implement and test all known algorithms for large data sets that are potentially practical. There algorithms include:
\begin{enumerate}
    \item Primal-dual \citep{CharikarKMN01} on our new coreset. 
    \item Passive uniform sampling followed by $k$-means++.
    \item Aggressive uniform sampling followed by $k$-means++.
    \item $k$-means++ on the original input.
    \item \algname followed by $k$-means++, on our new coreset. 
    \item Two local search algorithms by \citep{ChawlaG13,GuptaKLMV17}, on our new coreset. 
\end{enumerate}
We note that once these algorithms identify the $k$ centers, we consider all points and discard the farthest $z$ points. Our experiments show that our $\algname$ consistently outperforms all other algorithms. 
\fi

\subsection{Comparison to the Previous Work}

\paragraph{Algorithms for $k$-Means with Outliers:}

To understand the contribution of our work, it is important to contrast the algorithm with previous work. We believe a significant contribution of our work is the algorithmic simplicity and speed as well as the theoretical bounds that our approach guarantees. In particular, we will discuss why the previous algorithms are difficult to use in practice.

The first potentially practical algorithm developed is based on primal-dual \citep{CharikarKMN01}. Instead of solving a linear program (LP) and converting the solution to an integer solution, the primal-dual approach only uses the LP and its dual to guide the algorithm. However, the algorithm does not scale well and is not easy to implement.
In particular, it involves increasing variables uniformly, which requires $\Omega(n^2)$ running time and extra care to handle precision issues of fractional values. As mentioned before, this algorithm was never implemented prior to this paper. Our experiments show that this algorithm considerably under-performs compared to other algorithms. 

The second potentially practical algorithm is based on uniform sampling \citep{meyerson2004k}. The main observation of \cite{meyerson2004k} is that if every cluster is large enough, then a small uniform sample can serve as a coreset. This observation leads to two algorithms for $k$-means clustering with outliers: (i) (implicit) reduction to $k$-means clustering via conservative uniform sampling and (ii) (explicit) aggressive uniform sampling plus primal-dual \citep{CharikarKMN01}. In (i) it can be shown that a constant approximate $k$-means clustering of a uniform sample of size $n / (2z)$ is a constant approximation for $k$-means clustering with outliers, under the assumption that every cluster has size $\Omega(z \log k)$. Here, the main idea is to avoid any noise by sampling conservatively. Although this assumption is reasonable as discussed before, the real issue is that conservative uniform sampling doesn't give a sufficiently accurate sketch to be adopted in practice. For example, if there are $1\%$ noise points, then the conservative uniform sample has only $50$ points. In (ii),  a more aggressive uniform sampling is used and followed by the primal dual \citep{CharikarKMN01}. It first obtains a uniform sample of size $\Theta(k (n / z) \log n)$; then the (expected) number of outliers in the sample becomes $\Theta(k \log n)$. This aggressive uniform sampling turns out to have very little loss in terms of accuracy.
However, as mentioned before, the
 primal-dual algorithm under-performs compared to other algorithms in speed and accuracy. 

%It first obtains a uniform sample of size $\Theta(k (n / z) \log n)$; then the (expected) number of outliers in the sample becomes $k \log n$. This aggressive uniform sampling turns out to have very little loss in terms of accuracy. However, the performance considerably degrades in the subsequent execution of the primal-dual algorithm \citep{CharikarKMN01}.

Another line of algorithmic development has been based on local search \citep{ChawlaG13,GuptaKLMV17}. The algorithm in \cite{ChawlaG13} guarantees the convergence to  a local optimum, but has no approximation guarantees. The other algorithm \citep{GuptaKLMV17} is an $O(1)$-approximation but theoretically it may end up with discarding $O(k z\log n)$ outliers.  These local search algorithms are considerably slower than our method and the theoretical guarantees require discarding many more points.

%While these algorithms are easy to implement, we believe they will unlikely be widely used in practice, as they have been. After all, the history shows that for $k$-means clustering, $k$-means++, a variant of Lloyd's, has been the popular algorithm of choice over local search based algorithms, which are typically slower. 

To summarize, there is a need for a  fast and effective algorithm for $k$-means clustering with outliers. 

\paragraph{Coresets for $k$-Means with Outliers:}

The other main contribution of our work is a coreset for $k$-means with outliers of size $O(k \log n)$ - independent of the number of outliers $z$ and dimension $d$.

The notion of coreset we consider is related to the concept of a \emph{weak coreset} in the literature - see e.g. \cite{feldman2011coreset} for discussion of weak coresets and other types of coresets. Previous coreset constructions (some for stronger notions of coreset) have polynomial dependence on the number of outliers $z$ \citep{GuptaKLMV17}, inverse polynomial dependence on the fraction of outliers $\frac{z}{n}$ \citep{meyerson2004k, huang2018coreset}, or polynomial dependence on the dimension $d$ \citep{huang2018coreset}. Thus, all coresets constructed in the previous work can have large size for some value of $z$, e.g. $z = \Theta(\sqrt n)$, or for large values of $d$. In contrast, our construction is efficient for \emph{all} values of $z \in [0,n]$ and yields coresets of size with no dependence on $d$ or $z$.

%Importantly, our construction is efficient for \emph{all} values of $z \in [0,n]$, while constructions with polynomial dependence on $z$ have large size when $ z = \Omega(n)$, so when the number of outliers is a constant fraction of $n$, and constructions with inverse polynomial dependence on $\frac{z}{n}$ have large size when $z = o(n)$.

\subsection{Overview of Our Algorithms: \algname~ and \framename} Our preprocessing step, \algname, is reminiscent of density-based clustering. Our algorithm tags an input point as light if it has relatively few points around it. Formally, a point is declared as light if it has less than $2z$ points within a certain distance threshold $r$, which can be set by binary search. Then a point is discarded if it only has light points within distance $r$. We emphasize that the threshold is chosen by the algorithm, not by the algorithm user, unlike in density-based clustering. While our preprocessing step looks similar to the algorithm for $k$-center clustering \citep{CharikarKMN01}, which optimizes the $\ell_\infty$-loss, we find it surprising that a similar idea can be used for $k$-means clustering. 

It can take considerable time to label each point light or not. To speed up our algorithm, we develop a new corest construction for $k$-means with outliers. The idea is relatively simple. We first use aggressive sampling as in \cite{meyerson2004k}. The resulting sample has size $O(\frac{kn}{z} \log n)$ and includes $O(k \log n)$ outliers with high probability. Then we use $k$-means++ to obtain $O(k \log n)$ centers. As a result, we obtain a high-quality coreset of size $O(k \log n)$. Interestingly, to our best knowledge, combining aggressive sampling with another coreset for $k$-means with outliers has not been considered in the literature. 

\subsection{Other Related Work}

Due to the vast literature on clustering, 
we refer the reader to \cite{charu,kogan2006grouping,jainsurvey}  for an overview and survey of the literature. $k$-means clustering can be generalized by considering other norms of loss, and such extensions have been studied under different names. When the objective is $\ell_1$-norm loss, the problem is called $k$-medians. The $k$-median and $k$-mean clustering problems are closely related, and in general the algorithm and analysis for one can be readily translated into one for the other with an $O(1)$ factor loss in the approximation ratio. Constant approximations are known for $k$-medians and $k$-means based on linear programming, primal-dual, and local search \citep{arya2004local,charikar2002constant,charikar1999improved}. While its approximation ratio is $O(\log k)$, the $k$-means++ algorithm is widely used in practice for $k$-means clustering due to its practical performance and simplicity. When the loss function is $\ell_\infty$, the problem is known as $k$-centers and a $3$-approximation is known for $k$-centers clustering with outliers \citep{CharikarKMN01}. For recent work on these outlier problems in distributed settings, see~\cite{MalkomesKCWM15,li2018distributed,guha2017distributed,chen2018practical}.

\section{Preliminaries}

%\todo{I don't know if we need to state our problem in the general metric context.}
%In this paper we will consider the $k$-means clustering problem in a metric space. A metric space is defined as follows:\\

%\begin{definition}[Metric Space]
%	A metric space is a pair $(M,d)$ such that $M$ is a set and $d: M \times M \rightarrow \mathbb{R}$ is a map satisfying:
%	\begin{enumerate}[1)]
%		\item $d(x,y) \geq 0$ for all $x,y \in M$, and {\color{blue}$d(x,y) = 0 \iff x = y$}
%		\item $d(x,y) = d(y,x)$ for all $x,y \in M$
%		\item $d(x,z) \leq d(x,y) + d(y,z)$ for all $x,y,z \in M$ \quad (triangle inequality)
%	\end{enumerate}
%\end{definition}

%For the remainder of the paper, let $(M,d)$ be a metric space. Then for any $x \in M$ and $S \subset M$, we define the notation:
%\[d(x,S) \defeq \min\limits_{y \in S} d(x,y)\]

In this paper we will consider the Euclidean $k$-means with outliers problem as defined in the introduction.
%
%\begin{definition}[$k$-Means with Outliers]
%	In the $k$-means with outliers problem, we are given as input a subset $X = \{x_1, \dots, x_n\} \subset M$, a parameter $k \in \mathbb{N}$ (number of centers), and a parameter $z \in \mathbb{N}$ (number of outliers.)
%	The goal is to choose a collection of $k$ \textbf{centers}, $C = \{c_1, \dots c_k\} \subset M$, to minimize:
%	\[\sum\limits_{x \in X_z(C)} d^2(x,C)\]
%	, where $X_z(C) \subset X$ is the subset of $n-z$ input points with the smallest squared distances to $C$.
%	
%	More formally, without loss of generality, we may assume $d^2(x_1, C) \leq d^2(x_2, C) \leq \dots \leq d^2(x_n, C)$, where ties are broken arbitrarily but consistently. Then $X_z(C) = \{x_1, \dots, x_{n-z}\}$.
%\end{definition}
%
%\todo{we should remove this def, as it was introduced in intro.}
%
Note that the $\ell_2$-distance satisfies the \emph{triangle inequality}, so for all $x,y,z \in \mathbb{R}^d$, $d(x,z) \leq d(x,y) + d(y,z)$.
Further, the \emph{approximate triangle inequality} will be useful to our analyses (this follows from the triangle inequality):
%\[d^2(x,z) \leq 2d^2(x,y) + 2d^2(y,z) \quad \forall x,y,z \in \mathbb{R}^d\] % before compression
$d^2(x,z) \leq 2d^2(x,y) + 2d^2(y,z) \quad \forall x,y,z \in \mathbb{R}^d$.
Given a set of centers $C \subset \mathbb{R}^d$, we say that the \emph{assignment cost} of $x \in X$ to $C$ is $d^2(x, C)$. For $k$-means with outliers, a set, $C$, of $k$ centers naturally defines a clustering of the input points $X$ as follows:
\begin{definition}[Clustering]
	Let $C = \{c_1, \dots, c_k\} \subset \mathbb{R}^d$ be a set of $k$ centers. A clustering of $X$ defined by $C$ is a partition $C_1 \cup \dots \cup C_k$ of $X_z(C)$ satisfying:
	For all $x \in X_z$ and $c_i \in C$, $x \in C_i \iff d(x, C) = d(x, c_i)$, where ties between $c_i$'s are broken arbitrarily but consistently.
\end{definition}

In summary, for the $k$-means with outliers problem, given a set $C$ of $k$ centers, we assign each point in $X$ to its closest center in $C$. Then we exclude the $z$ points of $X$ with the highest assignment cost from the objective function (these points are our outliers.) This procedure defines a clustering of $X$ with outliers.

%\subsection{Notations} % before compression
\smallskip
\noindent
\textbf{Notations:}
For $n \in \mathbb{N}$, we define $[n] \defeq \{1, \dots, n\}$. Recall that as in the introduction, for any finite $Y \subset \mathbb{R}^d, x \in \mathbb{R}^d$, we define:
%\[d(x,Y) \defeq \max\limits_{y \in Y}d(x,y)\] % before compression
$d(x,Y) \defeq \min_{y \in Y}d(x,y)$.
For any $x \in \mathbb{R}^d,X \subseteq \mathbb{R}^d, r > 0$, we define the $X$-ball centered at $x$ with radius $r$ by
%\[B(x,r) \defeq \{y \in X \mid d(x,y) \leq r \}\] % before compression
$B(x,r) \defeq \{y \in X \mid d(x,y) \leq r \}$.
For a set of $k$ centers, $C \subset \mathbb{R}^d$, and $z \in \mathbb{N}$, we define the \emph{$z$-cost} of $C$ by 
%\[f_z^X(C) \defeq \sum\limits_{x \in X_z(C)} d^2(x,C).\]  % before compression
$f_z^X(C) \defeq \sum_{x \in X_z(C)} d^2(x,C).$
Recall that we define $X_z(C) \subset X$ to be the subset of points of $X$ excluding the $z$ points with highest assignment costs. Thus the $z$-cost of $C$ is the cost of clustering $X$ with $C$ while excluding the $z$ points with highest assignment costs.
As shorthand, when $z = 0$ -- so when we consider the $k$-means problem without outliers -- we will denote the $0$-cost of clustering $X$ with $C$ by $f^X(C) \defeq f^X_0(C)$.
Further, we say a set of $k$ centers $C^*$ is an \emph{optimal $z$-solution} if it minimizes $f_z^X(C)$ over all choices of $k$ centers, $C$. Then we define $Opt(X,k,z) \defeq f_z^X(C^*)$ to be the optimal objective value of the $k$-means with outliers instance $(X,k,z)$.
Analogously, for the $k$-means without outliers problem, we denote the optimal objective value of the $k$-means instance $(X,k)$ by $Opt(X,k)$.

\section{\algname \space Algorithm}

In this section, we will describe our  algorithm, $\algname$, which turns a $k$-means algorithm without outliers
to an algorithm for $k$-means with outliers in a black box fashion. We note that the algorithm naturally extends to $k$-medians with outliers and general metric spaces.
For the remainder of this section, let $X = \{x_1, \dots, x_n\} \subset \mathbb{R}^d$, $k \in \mathbb{N}$ , and $z \in \mathbb{N}$ define an instance of $k$-means with outliers.

\noindent\textbf{Algorithm Intuition:} The guiding intuition behind our algorithm is as follows:
We consider a ball of radius $r > 0$ around each point $x \in X$. If this ball contains many points, then $x$ is likely not to be an outlier in the optimal solution.

More concretely, if there are more than $2z$ points in $x$'s ball, then at most $z$ of these points can be outliers in the optimal solution. This means that the majority of $x$'s neighbourhood is real points in the optimal solution, so we can bound the assignment cost of $x$ to the optimal centers. We call such points \textit{heavy}.

There are $2$ main steps to our algorithm. First, we use the concept of heavy points to decide which points are real points and those that are outliers. Then we run a $k$-means approximation algorithm on the real points.

\noindent\textbf{Formal Algorithm:} Now we formally describe our algorithm $\algname$. As input, $\algname$ takes a $k$-means with outliers instance $(X,k,z)$ and an algorithm for $k$-mean without outliers, $\mathcal{A}$, where $\mathcal{A}$ takes an instance of $k$-means as input.

We will prove that if $\mathcal{A}$ is an $O(1)$-approximation for $k$-means and the optimal clusters are sufficiently large with respect to $z$, then $\algname$ outputs a good clustering that discards $O(z)$ outliers.
More precisely, we will prove the following theorem about the performance of $\algname$:
\begin{theorem}\label{algmain}
	Let $C$ be the output of $\algname(X, k ,z, \mathcal{A})$. Suppose that $\mathcal{A}$ is an $\alpha$-approximation for $k$-means. If every cluster in the clustering defined by $C^*$ has size at least $3z$, then $f_{2z}^X(C) \leq 9\alpha \cdot Opt(X,k,z)$.
\end{theorem}

\begin{corollary}
		Let $C$ be the output of $\algname(X, k ,z, \mathcal{A})$. Suppose that $\mathcal{A}$ is an $\alpha$-approximation. Then $f_{3kz + 2z}^X(C) \leq 9\alpha \cdot Opt(X,k,z)$.
\end{corollary}

In other words, \algname\space gives a pseudo-approximation-preserving reduction from $k$-means with outliers to $k$-means, where any $\alpha$ approximation for $k$-means implies a $9\alpha$ pseudo-approximation for $k$-means with outliers that throws away $3kz + 2z$ points as outliers.

%\begin{wrapfigure}{r}{0.50\textwidth}
%\begin{minipage}{0.50\textwidth}
\begin{figure}
\begin{algorithm}[H]
	\caption{for $k$-means with outliers} \label{alg:nk}
	$\algname(X, k, z, \mathcal{A})$
	\begin{algorithmic}[1]
		\STATE Suppose we know the optimal objective value $Opt \defeq Opt(X,k,z)$
		\STATE Initialize $r \gets 2(Opt/z)^{1/2}$, $Y \gets \emptyset$
		\FOR {each $x \in X$}\label{computeball}
			\STATE Compute $B(x,r)$
			\IF {$\lvert B(x,r) \rvert \geq 2z$}
				\STATE Mark $x$ as \textit{heavy}
			\ENDIF			
		\ENDFOR	
		\FOR {each $x \in X$}\label{computeheavy}
			\IF {$B(x,r)$ contains no heavy points}
				\STATE Update $Y \gets Y \cup \{x\}$
			\ENDIF
		\ENDFOR
		\STATE Output $C \gets \mathcal{A}(X \setminus Y, k)$
	\end{algorithmic}
\end{algorithm}
\end{figure}
%\end{minipage}
%\end{wrapfigure}

\subsection{Implementation Details}

Here we describe a simple implementation of $\algname$ that achieves runtime $O(n^2 d) + T(n)$ assuming we know the optimal objective value, $Opt$, where $T(n)$ is the runtime of the algorithm $\mathcal{A}$ on inputs of size $n$. This assumption can be removed by running that algorithm for many guesses of $Opt$, say by trying all powers of $2$ to obtain a $2$-approximation of $Opt$ for the correct guess.

For our experiments, we implement the loop in Line \ref{computeball} by enumerating over all pairs of points and computing their distance. This step takes time $O(n^2 d)$. We implement the loop in Line \ref{computeheavy} by enumerating over all elements in $B(x,r)$ and checking if it is heavy for each $x \in X$. This step takes $O(n^2)$. Running $\mathcal{A}$ on $(X \setminus Y, k)$ takes $T(n)$ time. We summarize the result of this section in the following lemma:

\begin{lemma}\label{lemruntime}
    Assuming that we know $Opt$ and that $\mathcal{A}$ takes time $T(n)$ on inputs of size $n$, then $\algname$ can be implemented to run in time $O(n^2 d) + T(n)$.
\end{lemma}

%\section{Framework for Big Data}

%\section{Coreset of Near Linear Size in $k$: Independent of the Number of Outliers}
\section{Coreset of Near Linear Size in k}

In this section we develop a general framework to speed up any $k$-means with outliers algorithm, and we apply this framework to\space \algname \space to show that we can achieve near-linear runtime. In particular, we achieve this by constructing what is called a \textit{coreset} for the $k$-means with outliers problem of size $O(k \log n)$, which is \emph{independent} of the number of outliers, $z$.

\subsection{Coresets for k-Means with Outliers}\label{sec:coreset-def}

Our coreset construction will leverage existing constructions of coresets for $k$-means with outliers. A coreset gives a good summary of the input instance in the following sense:

\begin{definition}[Coreset for $k$-Means with Outliers]\footnote{Note that our definition of coreset is parameterized by the \emph{number} of outliers, $z$, in contrast to previous work such as \cite{meyerson2004k} and \cite{huang2018coreset}, whose constructions are parametereized by the \emph{fraction} of outliers, $z/n$.}
	Let $(X,k, z)$ be an instance of $k$-means with outliers and $Y$ be a (possibly weighted) subset of $\mathbb{R}^d$. We say the $k$-means with outliers instance $(Y,k, z')$ is an $(\alpha, \beta)$-coreset for $X$ if for any set $C \subset \mathbb{R}^d$ of $k$-centers satisfying $f_{\kappa_1 z'}^Y(C) \leq \kappa_2 Opt(Y,k,z')$ for some $\kappa_1, \kappa_2 > 0$, we have $f_{\alpha \kappa_1 z}^X(C) \leq \beta \kappa_2 Opt(X,k,z)$.
\end{definition}

In words, if $(Y,k,z')$ is an $(\alpha, \beta)$ coreset for $(X,k,z)$, then running any $(\kappa_1, \kappa_2)$-approximate $k$-means with outliers algorithm on $(Y,k,z')$ (meaning the algorithm throws away $\kappa_1 z'$ outliers and outputs a solution with cost at most $\kappa_2 Opt(Y,k,z')$) gives a $(\alpha \kappa_1, \beta \kappa_2)$-approximate solution to $(X,k,z)$.

Note that if $Y$ is a weighted set with weights $w: Y \rightarrow \mathbb{R}_+$, then the $k$-means with outliers problem is analogously defined, where the objective is a weighted sum of assignment costs: $\min\limits_C \sum_{y \in Y_z(C)} w(y) d^2(y,C)$. Further, note that $\algname$ generalizes naturally to weighted $k$-means with outliers with the same guarantees.

%One should imagine that a coreset gives a good summary of the input instance. In fact, this summary is good enough such that it suffices to (approximately) solve the $k$-means with outliers problem on the coreset as opposed to the input set.

%The parameter $\kappa_1$ controls the approximation ratio of the algorithm used on the coreset and $\kappa_2$ controls the  factor outliers the algorithm needs to discard. %Note that $z' \leq z$ is used on the coreset because the number of outliers in the coreset maybe less than the original input.

The two coresets we will utilize for our construction are $\km$ \citep{AggarwalDK09} and Meyerson's sampling coreset \citep{meyerson2004k}. The guarantees of these coresets are as follows:

\begin{theorem}[\km]\label{kmcore}
	Let $\km(X,k)$ denote running $\km$ on input points $X$ to obtain a set $Y \subset X$ of size $k$. Further, let $Y_1, \dots Y_k$ be the clustering of $X$ with centers $y_1, \dots, y_k \in Y$, respectively. We define a weight function $w: Y \rightarrow \mathbb{R}_+$ by	
	$w(y_i) = \lvert Y_i \rvert$ for all $y_i \in Y$. Suppose $Y = \km(X, 32(k+z))$. Then with probability at least $0.03$, the instance $(Y, k, z)$ where $Y$ has weights $w$ is an $(1,124)$-coreset for the $k$-means with outliers instance $(X,k,z)$.
\end{theorem}

\begin{theorem}[Sampling]\label{samplecore}
	Let $S$ be a sample from $X$, where every $x \in X$ is included in $S$  independently with probability $p = \max(\frac{36}{z} \log (\frac{4nk^2}{z}), 36 \frac{k}{z} \log (2k^3))$.
	Then with probability at least $1 - \frac{1}{k^2}$, the instance $(S,k, 2.5 pz)$ is a $(16,29)$-coreset for $(X,k,z)$.
\end{theorem}

Observe that $\km$ gives a coreset of size $O(k+z)$, and uniform sampling gives a coreset of size $O(\frac{kn}{z} \log n)$ in expectation. If $z$ is small, then $\km$ gives a very compact coreset for $k$-means with outliers, but if $z$ is large -- say $z = \Omega(n)$ -- then $\km$ gives a coreset of linear size.
However, the case where $z$ is large is exactly when uniform sampling gives a small coreset.

In the next section, we show how we can combine these two coresets to construct a small coreset that works for all $z$.

\subsection{Our Coreset Construction:\space \framename}

%\begin{wrapfigure}{r}{0.60\textwidth}
%\begin{minipage}{0.60\textwidth}
\begin{figure}
\begin{algorithm}[H]
	%\caption{Framework for Big Data for $k$-Means with Outliers}
	\caption{Coreset Constuction for $k$-Means with Outliers}
	$\framename(X, k, z)$
	\begin{algorithmic}[1]
	    \STATE Let $p = \max(\frac{36}{z} \log (\frac{4nk^2}{z}), 36 \frac{k}{z} \log (2k^3))$.
	    \IF {$p > 1$}
	        \STATE Output $Y \gets \km(X,32(k + z))$.
	    \ELSE
	        \STATE Let $S$ be a sample drawn from $X$, where each $x \in X$ is included in $S$ independently with probability $p$.
		    \STATE Output $Y \gets \km(S, 32(k + 2.5pz))$
        \ENDIF
	\end{algorithmic}
	\label{alg:core}
\end{algorithm}
\end{figure}
%\end{minipage}
%\end{wrapfigure}

Using the above results, our strategy is as follows: Let $(X,k,z)$ be an instance of $k$-means with outliers. If $p > 1$, then we can show that $z = O(k \log n)$, so we can simply run \km\space on the input instance to get a good coreset. Otherwise, $z$ is large, so we first subsample approximately $\frac{kn}{z}$ points from $X$. Let $S$ denote the resulting sample. Then we compute a coreset on $S$ of size $32(k + 2.5 pz)$, where we scale down the number of outliers from $X$ proportionally.

Algorithm~\ref{alg:core} formally describes our coreset construction. We will prove that $\framename$ outputs with constant probability a good coreset for the $k$-means with outliers instance $(X,k,z)$ of size $O(k \log n)$. In particular, we will show:

\begin{theorem}\label{thmcoreset}
	With constant probability, $\framename$ outputs an $(O(1), O(1))$-coreset for the $k$-means with outliers instance $(X,k,z)$ of size $O(k \log n)$ in expectation.
\end{theorem}

\subsection{A Near Linear Time Algorithm for k-Means With Outliers}

Using $\framename$,  we show how to speed up $\algname$ to run in near linear time:
Let $Y$ be the result of $\framename(X, k, z)$. Then, to choose $k$ centers we run $\algname(Y,k,z, \mathcal{A})$ if $p > 1$; otherwise, run $\algname(Y,k, 2.5 pz, \mathcal{A})$, where $\mathcal{A}$ is any $O(1)$-approximate $k$-means algorithm with runtime $T(n)$ on inputs of size $n$.

\begin{theorem}\label{lineartime}
    There exists an algorithm that outputs with a constant probability an $O(1)$-approximate solution to $k$-means with outliers while discarding $O(kz)$ outliers in expected time $O(kd n \log^2 n) + T(k \log n)$.
\end{theorem}

\iffalse
\section{A Near Linear Time Algorithm for $k$-Means With Outliers}

In this section, we present our linear time algorithm for $k$-means with outliers based on \algname~ and \framename.

The following algorithm achieves near-linear running time and constant approximation factor. Let $\mathcal{A}$ be any $O(1)$ approximate $k$-means algorithm with runtime $T(n,k)$ for input $(X,k)$, where $X$ has size $n$.

\begin{algorithm}[H]
	\caption{Near Linear Time for $k$-Means with Outliers}
	$\textsc{Fast}\algname(X, k, z)$
	\begin{algorithmic}[1]
	    \STATE Initialize $Y \gets \framename(X, k, z)$
	    \IF {$z \leq k \log n$}
	        \STATE Initialize $C \gets \algname(Y,k,z)$
	    \ELSE
	        \STATE Initialize $C \gets \algname(Y,k, 2.5 \frac{z}{n}s)$
	    \ENDIF
		\STATE Output $C$
	\end{algorithmic}
\end{algorithm}
Combining Theorems~\ref{algmain} and~\ref{thmcoreset}, we have:

\begin{theorem}\label{lineartime}
    $\textsc{Fast}\algname$ outputs with constant probability an $O(1)$-approximate solution to $k$-means with outliers while discarding $O(kz)$ outliers in time $O(k^2 d n \log^2 n) + T(k \log n)$
\end{theorem}

\fi

\section{Experiment Results}

This section presents our experimental results. The main conclusions are:
\begin{compactitem}
\item Our algorithm \algname~ almost always has the best performance and finds the largest proportion of ground truth outliers. In the cases where \algname~ is not the best, it is competitive within 5\%.
\item Our algorithm results in a stable solution.  Algorithms without theoretical guarantees have unstable objectives on some experiments. 
\item Our coreset construction \framename~ allows us to run slower, more sophisticated, algorithms with theoretical guarantees on large inputs.  Despite their theoretical guarantees, their practical performance is not competitive.
\end{compactitem}

\begin{table*}[h]
\begin{center}
\begin{tabular}{c|c|c|c|c|c}
\hline
 & Primal-Dual & $k$-means-- & Local Search & Uniform Sample  & \algname\\
 \hline 
run time > 4hrs & 9/16 & 1/16 & 8/16 & 0/16 & 0/16  \\
 \hline
precision < 0.8 & 2/16 & 0/16 & 0/16 & 4/16 & 0/16  \\
 \hline 
total failure & 11/16 & 1/16 & 8/16 & 4/16 & 0/16 \\
\hline
\end{tabular}
\caption {Failure rates due to high run time or low precision.} \label{table:failure-main}
\end{center}
\end{table*}

The experiments shows that for a modest overhead for preprocessing, \algname~makes $k$-means clustering more robust to noise.

\noindent \textbf{Algorithms Implemented:} Our new coreset construction  makes it feasible to compare many algorithms for large data sets. Without this, most known algorithms for k-means with outliers become prohibitively slow even on modestly sized data sets. In our experiments, the coreset construction we utilize is \framename. More precisely, we first obtain a uniform sample by sampling each point independently with probability $p = \min \{\frac{2.5 k \log n}{z}, 1\}$. Then, we run $k$-means++ on the sample to choose $k + pz$ centers -- the resulting coreset is of size $k + pz$. 

Next we describe the algorithms tested. Besides the coreset construction, we use $k$-means++ to mean running $k$-means++ and then Lloyd's algorithm for brevity. For more details, see Supplementary Material~\ref{sec:app-exp}. In the following, ``on coreset'' refers to running the algorithm on the coreset as opposed to the entire input. For fair comparison, we ensure each algorithm discards \emph{exactly} $z$ outliers regardless of the theoretical guarantee.  At the end of each algorithm's execution, we discard the $z$ farthest points from the chosen $k$ centers as outliers.

\textbf{Algorithms Tested:}
\begin{compactenum}
    \item \textbf{\algname~ (plus $k$-means++ on coreset)}: We use \algname~with $k$-means++ as the input $\mathcal{A}$. The algorithm requires a bound on the objective $Opt$. For this, we considered powers of 2 in the range of $[n \min_{u, v \in X} d^2(u, v), n \max_{u, v \in X} d^2(u, v)]$. 
    \item \textbf{$k$-means++ (on the original input)}: Note this algorithm is not designed to handle outliers.  %The coreset is not used in this algorithm.
    \item \textbf{$k$-means++ (on coreset)}: Same note as the above. %This algorithm runs $k$-means++ on the coreset. 
    \item \textbf{Primal-dual algorithm of \cite{CharikarKMN01} (on coreset)}: A sophisticated algorithm based on constructing an approximate linear program solution.
    \item \textbf{Uniform Sample (conservative uniform sampling plus $k$-means++)}: We run  $k$-means++ on a uniform sample consisting of points sampled with probability $1 / (2z)$.
    %    A uniform sample was obtained by sampling each point with probability $1 / (2z)$, and then,  on the sample. 
    \item \textbf{$k$-means-- \citep{ChawlaG13} on coreset}: This algorithm is a variant of  the Lloyd's algorithm that executes each iteration of Lloyd's excluding the farthest $z$ points. 
    \item \textbf{Local search of \cite{GuptaKLMV17} (on coreset) }: This is an extension of the well-known $k$-means local search algorithm.%: In principle, this algorithm may end up with discarding $\Omega(z k \log n)$ points. However, it was observed that it never discarded more than $2z$ points in experimentation. We adopt the practical implementation of the algorithm described in \citep{GuptaKLMV17}. When the algorithm converges we enforce the farthest $z$ points to be the outliers. 
        %Therefore, we chose a practical implementation of the algorithm, which was running $k$-means++, discarding $z$ outliers, and then running $k$-means++ again on the remaining data. 
\end{compactenum}

    %We note that once these algorithms identify the $k$ centers, we consider all points and discard the farthest $z$ points. Our experiments show that our $\algname$ consistently outperforms all other algorithms

\iffalse. 
\subsection{Test Data Sets}

We used the following data sets in our experiments. 

\begin{enumerate}
    \item Gauss.
    \item \texttt{susy-}$\Delta$\citep{susy}. This data set has been generated by  Monte Carlo simulations described in \citep{baldi2014searching}. The set has 5M instances, each having 18 features. We normalized the data such that the average and std are 0 and 1 on each dimension, respectively.  Then, we randomly sampled $z = 1\% * 5M$ points uniformly at random from $[-\Delta, \Delta]^{18}$ and added as noise. 
    \item
    \item 
\end{enumerate}
\fi

\noindent \textbf{Experiments:} We now describe our experiments which were done on both synthetic and real data sets. 

\paragraph{Synthetic Data Experiments} We first conducted experiments with synthetic data sets of various parameters. Every data set has $n$ equal one million points and $k, d \in \{10, 20\}$ and $z \in \{10000, 50000\}$.  Then we generated $k$ random Gaussian balls. For the $i$th Gaussian we choose a center $c_i$ from $[-1/2, 1/2]^d$ uniformly at random. These are the true centers. Then, we add $n/k$ points drawn from $\mathcal{N}(c_i,1 )$ for the $i$th Gaussian. Next, we add noise. Points that are noise were sampled uniformly at random either from the same range $[-1/2, 1/2]^d$ or from a larger range $[-5/2, 5/2]^d$ depending on the experiment. We tagged the farthest $z$ points from the centers $\{c_1, \ldots, c_k\}$ as ground truth outliers. We consider all possible $16$  combinations of $k,d, z$ values and the noise range. 

Each experiment was conducted $3$ times, and we chose the result with the minimum objective and measured the total running time over all $3$ runs.   We aborted the execution if the algorithm failed to terminate within $4$ hours. %All experiments were performed on MERCED cluster\footnote{We gratefully acknowledge computing time on the Multi-EnvironmentComputer for Exploration and Discovery (MERCED) cluster at UC Merced, which was funded by National Science Foundation Grant No. ACI-1429783.} using a single node with 20 cores at 2301MHz and RAM size 128GB.
All experiments were performed on a cluster using a single node with 20 cores at 2301MHz and RAM size 128GB. Table \ref{table:failure-main} shows the number of times each algorithm aborted due to high run time. Also we measured the recall, which is defined as number of ground truth outliers reported by the algorithm, divided by $z$, the number of points discarded. The recall was the same as the precision in all cases, so we use precision in the remaining text. We choose $0.8$ as the threshold for the acceptable precision and counted the number of inputs for which each algorithm had precision lower than $0.8$. Our algorithm \algname, $k$-means++ on coreset, and  $k$-means++ on the original input all had precision greater than $0.99$ for all data sets and always terminated within $4$ hours. The $k$-means++ results are excluded from the table. Details of the quality and runtime are deferred to the Supplementary Material  ~\ref{sec:app-exp}.

\begin{table*}[h]
\begin{center}
\begin{tabular}{c|c|c|c|c|c|c|c}
\hline
                &\textsc{Skin}-5 & \textsc{Skin}-10 & \textsc{Susy}-5 & \textsc{Susy}-10 & \textsc{Power}-5 & \textsc{Power}-10 & \textsc{KddFull}
 \\

 \hline 
  \multirow{3}{*}{\algname} & 1&\bf{1}&\bf{1}&\bf{1}&\bf{1}&\bf{1}&\bf{1}\\&\bf{0.8065}& \bf{0.9424} & 0.8518 & 0.9774 & 0.6720 & 0.9679 & 0.6187\\&
  56 & 56& 1136 & 1144 & 363 & 350 & 1027

  \\
 \hline
   \multirow{3}{*}{$k$-means--} & 0.9740 & 1.5082 & 1.2096 & 1.1414 & 1.0587 & 1.0625 & 2.0259\\&
   0.7632 & 0.9044 & 0.8151 & 0.9753 & 0.6857 & 0.9673 & \bf{0.6436} \\&
   86 & 89 & 672 & 697 & 291 & 251 & 122
 \\
 \hline 
    \multirow{3}{*}\textsc{$k$-means++} &  1.0641 & 1.4417 & 1.0150 & 1.0091 & 1.0815 & 1.0876 & 1.5825 \\coreset&0.7653 & 0.9012 & \bf{0.8622} & \bf{0.9865} & \bf{0.7247} & \bf{0.9681} & 0.3088 \\& 39 & 37 & 462 & 465 & 177 & 142 & 124

  \\
\hline
 
   \multirow{3}{*}\textsc{$k$-means++} & \bf{0.9525} & 1.6676 & 1.0017 & 1.0351 & 1.0278 & 1.0535 & 1.5756 \\original& 0.7775 & 0.8975 & 0.8478 & 0.9814 & 0.7116 & 0.9649 & 0.3259 \\&
   34 & 43 & 6900 & 6054 & 689 & 943 & 652
\\
\hline

\end{tabular}
\end{center}
\caption {Experiment results on real data sets with $\Delta = 5, 10$. The top, middle, bottom in each entry are the objective (normalized relative to \algname), precision, and run time (sec.), resp. Bold indicates the best in the category.}
\label{table:real}
\end{table*}

\iffalse
\begin{table}[]
\begin{tabular}{c|c|c|c|c|c|c|c}
 & \algname & (P)$k$-means++ & $k$-means++ & PD & $k$-means-- & LS & Uniform  \\
 \hline
RT > 4hrs & 0 & 0 &  &  &  &  &  \\
 \hline
Precision < 0.8 &  &  &  &  &  &  & 
\end{tabular}
\end{table}
\fi

\paragraph{Real Data Experiments} For further experiments, we used real data sets. We used the same normalization, noise addition method and the same value of $k = 10$ in all experiments.  The data sets are \textsc{Skin}-$\Delta$, \textsc{Susy}-$\Delta$, and \textsc{Power}-$\Delta$. We normalized the data such that the mean and standard deviation are $0$ and $1$ on each dimension, respectively.  Then we randomly sampled $z = 0.01 n$ points uniformly at random from $[-\Delta, \Delta]^{d}$ and added them as 
noise. We discarded data points with missing entries.

\iffalse
\begin{enumerate}
    \item \textsc{Skin}-$\Delta$ \citep{skin}. $n = 245057$, $d = 3$, $k = 10$, $z = 0.01 n$. Only the first three features were used. 
    \item \textsc{Susy}-$\Delta$ \citep{susy}. $n = 5$M, $d = 18$, $k = 10$, $z = 0.01 n$. 
    \item \textsc{Power}-$\Delta$ \citep{power}. $n = 2049280$, $d = 7$, $k = 10$, $z = 0.01n$. Out of 9 features, we dropped the first two, date and time, that denote when the measurements were made. 
    \item \textsc{KddFull} \citep{kddfull}. $n = 4898431$, $d = 34$, $k = 3$, 
$z = 45747$. Each instance has $41$ features and we excluded 7 non-numeric features. This data set has 23 classes and 3 classes %(normal, nepture, smurf) 
account for 98.3\% of the instances. We considered the other 45747 instances as ground truth outliers. 
\end{enumerate}
\fi

\textbf{Real Data Sets:}
\begin{compactenum}
    \item \textbf{\textsc{Skin}-$\Delta$} \citep{skin}. $n = 245057$, $d = 3$, $k = 10$, $z = 0.01 n$. Only the first $3$ features were used.
    \item \textbf{\textsc{Susy}-$\Delta$} \citep{susy}. $n = 5$M, $d = 18$, $k = 10$, $z = 0.01 n$.
    \item \textbf{\textsc{Power}-$\Delta$} \citep{power}. $n = 2049280$, $d = 7$, $k = 10$, $z = 0.01n$. Out of $9$ features, we dropped the first $2$, date and time, that denote when the measurements were made. 
    \item \textbf{\textsc{KddFull}} \citep{kddfull}. $n = 4898431$, $d = 34$, $k = 3$,  $z = 45747$. Each instance has $41$ features and we excluded $7$ non-numeric features. This data set has $23$ classes and $3$ classes account for $98.3\%$ of the data points. We considered the other $45747$ data points as ground truth outliers.  
\end{compactenum}

% Thus, unlike in the experiment results of the synthetic data sets, a clustering with higher precision is not necessarily better than others with lower precision.  his is because we do not have absolute ground truth outliers.

%Note that some of the random noise points added could be considered as non-outliers if they are close to the underlying clusters. Therefore, we believe the objective is a more accurate measure of the solution quality. 

Table~\ref{table:real} shows our experiment results for the above real data sets. Due to their high failure rate observed in Table~\ref{table:failure-main} and space constraints, we excluded the primal-dual, local search, and  conservative uniform sampling algorithms from  Table~\ref{table:real}; all results can be found in Supplementary Material~\ref{sec:app-exp}.
%$We continued to test $k$-means-- since it failed to terminate only for one input. 
As before, we executed each algorithm $3$ times.  It is worth noting that \algname~is the \emph{only} algorithm with the worst case guarantees shown in Table~\ref{table:real}. This gives a candidate explanation for the stability of our algorithm's solution quality across all data sets in comparison to the other algorithms considered.

%It is worth noting \algname~ considerably outperformed other algorithms for the same data set with different $\Delta = 10$.As discussed above,  precision may not be the best measure here as some added noise points may not serve as noise. For this data set, \algname~ has the best objective and its precision is almost as high as others when $\Delta = 10$.

The result shows that our algorithm  \algname~has the best objective for all data sets, except within $5\%$ for \textsc{Skin}-5.  Our algorithm is always competitive with the best precision.  For \textsc{KddFull} where we didn't add artificial noise, $\algname$ significantly outperformed other algorithms in terms of objective. We can see that \algname~pays extra in the run time to remove outliers, but this preprocessing enables stability, and competitive performance.

\section{Conclusion}

This paper presents a near linear time algorithm for removing noise from data before applying a $k$-means clustering.  We show that the algorithm has provably strong guarantees on the number of outliers discarded and approximation ratio. Further, \algname\space gives the first pseudo-approximation-preserving reduction from $k$-means with outliers to $k$-means without outliers. Our experiments show that the algorithm is the fastest among algorithms with provable guarantees and is more accurate than state-of-the-art algorithms. It is of interest to determine if the algorithm achieves better guarantees if data has more structure such as being in low dimensional Euclidean space or being assumed to be well-clusterable \citep{BravermanMORST11}.

\section{Acknowledgments}

S. Im and M. Montazer Qaem were supported in part by NSF grants CCF-1617653 and 1844939. B. Moseley and R. Zhou  were supported in part by NSF grants CCF-1725543, 1733873, 1845146, a Google Research Award, a Bosch junior faculty chair and an Infor faculty award. 

%-------------------------------------------
\bibliographystyle{plainnat}
\bibliography{koutlier}
%-------------------------------------------
\clearpage

\appendix
\section*{Supplementary Material for: Fast Noise Removal for $k$-Means Clustering}

\section{Analysis of \space \algname}

The goal of this section is to prove Theorem \ref{algmain}.
For the remainder of this section, let $C^*$ denote the optimal $z$-solution and $C$ denote the output of $\algname(X,k,z, \mathcal{A})$. Again, let  $Opt \defeq Opt(X,k,z)$.

We first show the benefits of optimal clusters having size at least $3z$.

\begin{claim}\label{heavyprop}
   For each optimal center $c^* \in C^*$, let $x(c^*) \in X$ be the closest input point to $c^*$. If the cluster defined by $c^* \in C^*$ has size at least $3z$, then $d(x(c^*), c^*) \leq (\frac{Opt}{3z})^{1/2}$.
\end{claim}
\begin{proof}
    Assume for contradiction that $d(x(c^*), c^*) > (\frac{Opt}{3z})^{1/2}$. Thus for each input point $x \in X$ that is assigned to center $c^*$ in the optimal solution, we have $d^2(x,c^*) > \frac{Opt}{3z}$. There are at least $3z$ such points, so we can lower bound the assignment cost of these points by $3z(\frac{Opt}{3z}) = Opt$. This is a contradiction.
\end{proof}

\begin{lemma}\label{assume}
    If the cluster defined by $c^* \in C^*$ in the optimal solution has size at least $3z$, then $x(c^*)$ is heavy.
\end{lemma}
\begin{proof}
	Assume for contradiction that $x(c^*)$ is light, so $\lvert B(x(c^*), r) \rvert < 2z$. However, at least $3z$ points are assigned to $c^*$ in the optimal solution, so there are at least $z + 1$ such points outside of $B(x(c^*),r)$.
	
	Let $x \notin B(x(c^*),r)$ be such a point that is assigned to $c^*$ in the optimal solution. By the triangle inequality, we have:
	\[d(x,x(c^*)) \leq d(x,c^*) + d(x(c^*), c^*)\]
	, which implies $d(x,c^*) \geq r - d(x(c^*),c^*) \geq 2(Opt/z)^{1/2} - (1/3)^{1/2} (Opt/z)^{1/2} \geq (Opt/z)^{1/2}$.
	
	We conclude that for at least $z+1$ points assigned to $c^*$ in the optimal solution, their assignment costs are each at least $Opt/z$. This is a contradiction.
\end{proof}

Now using this result, we can upper bound the number of outliers required by $\algname(X,k,z, \mathcal{A})$ to remain competitive with the optimal $z$-solution (we will show that this quantity is upper bounded by the size of $Y$ at the end of $\algname(X,k,z, \mathcal{A})$.)\\

\begin{lemma}\label{numoutliers}
	At the end of $\algname(X,k,z, \mathcal{A})$, $\lvert Y \rvert \leq 3z \#\{\text{optimal clusters of size less than $3z$}\} + 2z$.
\end{lemma}
\begin{proof}
	Let $C^* = \{c_1^*, \dots, c_k^*\} \subset X$.
	
	For each $x \in X$, we will classify points into two types:
	\begin{enumerate}[1)]
		\item $d(x,C^*) \leq (Opt/z)^{1/2}$:
		
		We have that $x$ satisfies $d(x, C^*) = d(x, c^*) \leq (Opt/z)^{1/2}$ for some $c^* \in C^*$. If the cluster defined by $c^*$ has size at least $3z$, then by Lemma \ref{assume}, $x(c^*)$ is heavy.
		
		Further, $d(x, x(c^*)) \leq d(x, c^*) + d(x(c^*),c^*) \leq (Opt/z)^{1/2} + (1/3)^{1/2}(Opt/z)^{1/2} \leq 2 (Opt/z)^{1/2}$, so $x(c^*) \in B(x,r)$. Thus, we will not add $x$ to $Y$ if its nearest optimal cluster has size at least $3z$.
		
		\item $d(x,C^*) > (Opt/z)^{1/2}$:
		
		We claim that there are at most $2z$ such $x \in X$ satisfying $d(x,C^*) > (Opt/z)^{1/2}$. Assume for contradiction that there are at least $2z + 1$ points $x \in X$ with $d(x,C^*) > (Opt/z)^{1/2}$. At most $z$ of these points can be outliers, so the optimal solution must cluster at least $z+1$ of these points. Thus we can lower bound the assignment cost of these points to $C^*$ by:
		\[(z + 1)r^2 = (z + 1) (Opt/z) > Opt\]
		This is a contradiction.
	\end{enumerate}

	We conclude that $Y$ includes no points of type 1 from clusters of size at least $3z$, at most $3z$ points from each cluster of size less than $3z$, and at most $2z$ points of type 2.
\end{proof}
\begin{corollary}\label{corlarge}
    If every optimal cluster has size at least $3z$, then at the end of $\algname(X,k,z, \mathcal{A})$, $\lvert Y \rvert \leq 2z$.
\end{corollary}

It remains to bound the $\lvert Y \rvert$-cost of $C$. Recall that the $\lvert Y \rvert$-cost of $C$ is the cost of clustering $X$ with $C$ excluding the $\lvert Y \rvert$ points of largest assignment cost.

Intuitively, we do not need to worry about the points in $X$ that are clustered in both the $\lvert Y \rvert$-solution $C$ and the $z$-solution $C^*$ -- so the points in $X_{\lvert Y \rvert}(C) \cap X_z(C^*)$, because such points are paid for in both solutions.

We must take some care to bound the cost of the points in $X$ that are clustered by the $\lvert Y \rvert$-solution $C$ but are outliers in the $z$-solution $C^*$, because such points could have unbounded assignment costs to $C^*$. Here we will use the following property of heavy points:\\

\begin{lemma}\label{heavy}
	Let $x \in X$ be a heavy point. Then there exists some optimal center $c^* \in C^*$ such that $d(x, c^*) \leq 2r$.
\end{lemma}
\begin{proof}
	Assume for contradiction that $d(x, c^*) > 2r$ for every $c^* \in C^*$.
	However, $x$ is heavy, so $\lvert B(x, r) \rvert \geq 2z$. At least $z$ points in $B(x,r)$ must be clustered by the optimal $z$-solution $C^*$.
	
	Consider any such $x' \in B(x,r) \cap X_z(C^*)$. By the triangle inequality, we have 
	\[2r < d(x, C^*) \leq d(x, x') + d(x', C^*) \leq r + d(x',C^*)\]
	This implies $d(x', C^*) > r$.
	Thus we can lower bound the assignment cost to $C^*$ of all points in $B(x,r) \cap X_z(C^*)$ by:
	\[\sum\limits_{x' \in B(x,r) \cap X_z(C^*)} d^2(x', C^*) > zr^2 = 4Opt\]
	This is a contradiction.
\end{proof}

Now we are ready to prove the main theorem of this section.

\begin{proof}[Proof of Theorem \ref{algmain}]
	By Corolloary \ref{corlarge}, we have $f_{2z}^X(C) \leq f^{X \setminus Y}(C)$.
	
	Further, by construction, $C$ is an $\alpha$-approximate $k$-means solution on $X \setminus Y$. Then
	\[f^{X \setminus Y}(C) \leq \alpha f^{X \setminus Y^*} \leq \alpha f^{X \setminus Y}(C^*),\]
	so it suffices to show that $f^{X \setminus Y}(C^*) \leq 9 \cdot Opt$.
	
	We will consider two types of points:
	\begin{enumerate}[1)]
		\item $x \in (X \setminus Y) \cap X_z(C^*)$, so points in $X \setminus Y$ that are also clustered in the optimal $z$-solution $C^*$:
		
		We have \begin{eqnarray*}
		&&\sum\limits_{x \in (X \setminus Y) \cap X_z(C^*)} d^2(x,C^*) \\
        &&\leq \sum\limits_{x \in X_z(C^*)} d^2(x, C^*) = f_z^X(C^*).
		\end{eqnarray*}
		
		\item $x \in (X \setminus Y) \cap (X \setminus X_z(C^*))$, so points in $X \setminus Y$ that are outliers in the optimal $z$-solution $C^*$:
		
		Observe that by definition, $\lvert X \setminus X_z(C^*) \rvert = z$, so there are at most $z$ such $x$. By Lemma \ref{assume}, for each such $x \in (X \setminus Y) \cap (X \setminus X_z(C^*))$, we have $d^2(x, C^*) \leq 4r^2$.
		Thus, \[\sum\limits_{x \in (X \setminus Y) \cap (X \setminus X_z(C^*))} d^2(x,C^*) \leq z(4r^2) = 8 f_z^X(C^*).\]
	\end{enumerate}

	We conclude that $f^{X \setminus Y}(C^*) \leq f_z^X(C^*) + 8f_z^X(C^*) = 9 f_z^X(C^*)$, as required.
\end{proof}

\section{Analysis of Coreset Construction and Near Linear Time Algorithm}

The goal of this section is to prove Theorems \ref{thmcoreset} and \ref{lineartime}. In our proof, we will use Theorems \ref{kmcore} and \ref{samplecore}. For proofs of these theorems, see Sections \ref{kmproof} and \ref{sampleproof}.

\begin{proof}[Proof of Theorem \ref{thmcoreset}]
    We consider $2$ cases: $p > 1$ and $p \leq 1$.
    
    If $p > 1$, then $Y = \km(X,32(k + z))$. Because $p >1$, we have $\max(36\log(\frac{4nk^2}{z}), 36k \log (2k^3)) > z \Rightarrow z = O(k \log n)$. Then $\lvert Y \rvert = O(k + z) = O(k\log n)$, as required. Further, by Theorem \ref{kmcore}, $(Y,k,z)$ is a $(1, 124)$-coreset for $(X,k,z)$ with constant probability.
    
    Otherwise, if $p \leq 1$, then $Y = \km(S, 32(k + 2.5pz))$. Thus, $\lvert Y \rvert = O(k +pz) = O(k \log n)$, as required. By Theorem \ref{samplecore}, with probability at least $1 - \frac{1}{k^2}$, $(S,k,2.5 pz)$ is an $(16, 29)$-coreset for $(X,k,z)$. For the remainder of this analysis, we assume this condition holds. We also know that $(Y,k, 2.5 pz)$ is a $(1, 124)$-coreset for $(S,k, 2.5 pz)$ with constant probability. Assume this holds for the remainder of the analysis.
    
    Let $C$ be a set of $k$ centers satisfying $f^Y_{2.5 \kappa_1 pz}(C) \leq \kappa_2 Opt(Y,k, 2.5 pz)$.
	Because $(Y,k, 2.5 pz)$ is an $(1, 124)$-coreset for $(S,k, 2.5pz)$, this implies:
	\[f^S_{2.5\kappa_1 pz}(C) \leq 124 \kappa_2 Opt(S,k,2.5 pz)\]
	Because $(S,k, 2.5pz)$ is an $(16, 29)$-coreset for $(X,k,z)$, we conclude:
	\[f^X_{16 \kappa_1 z}(C) \leq 29 \cdot 124 \kappa_2 Opt(X,k,z)\]
	Thus $(Y,k, 2.5 pz)$ is an $(O(1), O(1))$-coreset for $(X,k,z)$.
\end{proof}

\begin{proof}[Proof of Theorem \ref{lineartime}]
    The approximation guarantees follow directly from Theorems~\ref{algmain} and~\ref{thmcoreset}.
    
    To analyze the runtime, note that we can compute $S$ in time $O(n)$. It is known that $\km(X,k)$ takes $O(kdn)$ time \citep{ArthurV07,AggarwalDK09}. Thus the runtime of $\framename$ is dominated by the runtime of $\km$ in both cases when $p >1$ and $p \leq 1$, which takes $O((k \log n) dn)$ time.
   
    Note that $Y$ has size $O(k \log n)$ in expectation, so by Lemma \ref{lemruntime}, $\algname$ can be implemented to run in time $O(k^2d \log^2 n) + T(k \log n)$ on $Y$ in expectation.
\end{proof}

\section{Proof of Theorem \ref{kmcore}}\label{kmproof}

Our proof of Theorem \ref{kmcore} relies on the following lemma which is implicit in \cite{AggarwalDK09}:

\begin{lemma}
	Let $Y = \km(X,32k)$. Then $f^X(Y) \leq 20 \cdot Opt(X,k)$ with probability at least $0.03$.
\end{lemma}

\begin{corollary}
	Let $Y = \km(X,32(k + z))$. Then $f^X(Y) \leq 20 \cdot Opt(X,k,z)$ with probability at least $0.03$.
\end{corollary}
\begin{proof}
	Let $C^*$ be the optimal solution to the $k$-means with outliers instance $(X,k,z)$. Note that $X \setminus X_{z}(C^*)$ is the set of outliers in the optimal solution, so $\lvert X \setminus X_z(C^*) \rvert \leq z$.
	
	Then we have $f^X_z(C^*) \geq f^X(C^* \cup (X \setminus X_z(C^*)) \geq Opt(X,k+z)$. Combining this inequality with the above lemma gives the desired result.
\end{proof}

Using the above corollary, we can prove Theorem \ref{kmcore} by a moving argument:

\begin{proof}[Proof of Theorem \ref{kmcore}]
	Let $Y = \km(X, 32(k+z))$ with weights $w$ as defined in the theorem statement. By the above corollary, we have $f^X(Y) \leq 20 \cdot Opt(X,k,z)$ with constant probability. We assume for the remainder of the proof that this condition holds.
	
	Let $C$ be any set of $k$ centers such that $f^Y_{\kappa_1 z}(C) \leq \kappa_2 Opt(Y,k,z)$. We wish to bound $f^X_{\kappa_1 z}(C)$.
	
	Note that by definition of $w$, $\sum\limits_{y \in Y} w(y) = n$, and each weight is an integer. Thus for the remainder of the proof we interpret $Y$ as a multiset such that for each $y \in Y$, there are $w(y)$ copies of $y$ in the multiset.
	
	It follows, we can associate each $x \in X$ with a unique $y(x) \in Y$ such that $d^2(x, y(x)) = d^2(x, Y)$ (so $y(x)$ is a unique copy of the center that $x$ is assigned to in the clustering of $X$ with centers $Y$.)
	
	Now we partition $X_{\kappa_1 z}(C)$ into two sets:
	\[X' \defeq \{x \in X \mid x \in X_{\kappa_1 z}(C), y(x) \in Y_{\kappa_1 z}(C)\}\]
	\[X'' \defeq \{x \in X \mid x \in X_{\kappa_1 z}(C), y(x) \notin Y_{\kappa_1 z}(C)\}\]
	
	For each $x \in X_{\kappa_1 z}(C)$, we want to bound its assignment cost.  There are two cases:
	\begin{enumerate}[1)]
		\item $x \in X'$:
		
		We can bound $d^2(x,C) \leq 2d^2(x, y(x)) + 2d^2(y(x), C)$. Note that $y(x) \in Y_{\kappa_1 z}(C)$, so we can bound the assignment cost $d^2(y(x), C)$.
		
		\item $x \in X''$:
		
		Note that because $y(x) \notin Y_{\kappa_1 z}(C)$, and $X_{\kappa_1 z}(C), Y_{\kappa_1 z}(C)$ are the same size, we can associate $y(x)$ with a unique element in $Y_{\kappa_1 z}(C)$, say $y(x') \in Y_{\kappa_1 z}(C)$ such that $x' \notin X_{\kappa_1 z}(C)$.
		
		Note that $x'$ is not assigned in the $\kappa_1 z$-solution $C$, but $x$ is assigned, so we can bound:
		\[d^2(x,C) \leq d^2(x', C) \leq 2d^2(x', y(x')) + 2d^2(y(x'), C)\]
	\end{enumerate}

	By summing over all $x \in X_{\kappa_1 z}(C)$ and applying the above bounds, we have:
	\begin{align*}
		f^X_{\kappa_1 z}(C) &= \sum\limits_{x \in X'} d^2(x,C) + \sum\limits_{x \in X''} d^2(x,C)\\
		&\leq 2\sum\limits_{x \in X'} (d^2(x,y(x)) + d^2(y(x), C))\\
		&\qquad + 2\sum\limits_{x \in X''} (d^2(x', y(x')) + d^2(y(x'),C))\\
		&=2f^X(Y) + 2f^Y_{\kappa_1 z}(C)\\
		&\leq 2 \cdot 20 \cdot Opt(X,k,z) + 2 \kappa_2 Opt(Y,k,z)
	\end{align*}
	An analogous argument gives that $Opt(Y,k,z) \leq 2 \cdot 20 \cdot Opt(X,k,z) + 2 \cdot Opt(X,k,z)$.
	
	We conclude that $f^X_{\kappa_1 z}(C) \leq (40 + 84\kappa_2)Opt(X,k,z)$, where we may assume $\kappa_2 \geq1$. This gives the desired result.
\end{proof}

\section{Proof of Theorem \ref{samplecore}}\label{sampleproof}

The proof of Theorem \ref{samplecore} closely follows \cite{meyerson2004k} and is given here for completeness. Note that the key difference is that rather than sampling elements uniformly from  $X$ with replacement as in \cite{meyerson2004k}, instead we sample each element of $X$ independently with probability $p$. In this section let $C^*=\{c_1^*, \dots, c_k^*\}$ be an optimal $z$-solution on $X$ with clusters $C_1^*, \dots, C_k^*$ such that $C_i^*$ is the set of points assigned to center $c_i^*$. Let $n_i \defeq \lvert C_i^* \rvert$ denote the size of cluster $i$.

Further, let $S$ be a sample drawn from $X$ of size $s$ as in $\framename$, and let $C$ a set of $k$ centers satisfying $f^S_{2.5\kappa_1 pz}(C) \leq \kappa_2 Opt(S,k, 2.5 pz)$ for some constants $\kappa_1, \kappa_2 > 0$.

The goal of this section is to prove that the sample $S$ gives a good coreset of $X$ for the $k$-means with outliers problem. We begin with some definitions that will be useful to our analysis:

\begin{definition}[Large/Small Clusters]
	We say a cluster $C_i^*$ of the optimal solution is \textbf{large} if $\lvert C_i^* \rvert \geq \frac{z}{k}$ and \textbf{small} otherwise.
\end{definition}

\begin{definition}[Covered/Uncovered Clusters]
	Let $A \defeq S_{2.5pz}(C^*) \cap S_{2.5\kappa_1 pz}(C)$, so $A$ is the set of points in $S$ that are assigned in the $2.5pz$-solution $C^*$ and in the $2.5\kappa_1 pz$-solution $C$.
	
	We say a large cluster $C_i^*$ is \textbf{covered} if $\lvert C_i^* \cap A \rvert \geq \frac{1}{2} \lvert C_i^* \cap S \rvert$ and \textbf{uncovered} otherwise.
\end{definition}

Intuitively, in our analysis we want to show that most of the large clusters are covered, because the large clusters make up the majority of the points. In order to obtain a good summary of the whole point set, it suffices to obtain a good summary of the large clusters.

We quantify this notion of a good summary by defining a division of $X$ into bins with respect to the centers, $C^*$.

\begin{definition}[Bin Division]
	Let $b \in \mathbb{N}$. The $b$-bin division of $X$ with respect to a set of $k$ centers, $C$, is a partition $B_1, \dots B_b$ of $X$ such that $B_1$ contains the $\frac{n}{b}$ points in $X$ with the smallest assignment costs to $C^*$, $B_2$ contains the next $\frac{n}{b}$ cheapest points, and so on. More formally, the partition $B_1, \dots, B_b$ satisfies:
	
	\begin{itemize}
		\item $\lvert B_i \rvert = \frac{n}{b}$ for all $i \in [b]$
		\item $\max\limits_{x \in B_i} d^2(x, C) \leq \min\limits_{x \in B_{i+1}} d^2(x,C)$ for all $i \in [b-1]$
	\end{itemize}
\end{definition}

For the remainder of this section, let $B_1, \dots, B_b$ denote the bin division of $X$ with respect to the optimal $z$-solution $C^*$, where $b = \frac{n}{z}$ (so each bin has size $z$.)

The following lemma shows that our sample size is sufficiently large to obtain a good representation of each bin and each large cluster.

\begin{lemma}\label{lemmacond}
	With probability at least $1 - \frac{1}{k^2}$, the following both hold:
	\begin{enumerate}[1)]
		\item For all $i \in [b]$, $\lvert S \cap B_i \rvert \in [0.75 pz, 1.25 pz]$\label{condbin}
		\item For every large cluster $C_i^*$, $\lvert S \cap C_i^* \rvert \geq 0.75 pn_i$\label{condlarge}
	\end{enumerate}
\end{lemma}
\begin{proof}
    We will use the following standard Chernoff bounds, where $X = \sum\limits_{ i \in [n]} X_i$ is the sum of $n$ i.i.d. random variables $X_i \sim Ber(p)$. For any $\delta \geq 0$, $Pr(\lvert X - pn \rvert \geq \delta pn) \leq 2exp(\frac{-\delta^2 pn}{2 + \delta})$ and $Pr(X \geq (1 - \delta) pn) \leq exp(\frac{-\delta^2 pn}{2 + \delta})$. 
    
    We bound the failure probability for each bin and each large cluster. For all $i \in [b]$, we have $Pr( \lvert \lvert S \cap B_i\rvert - pz \rvert \geq \frac{1}{4} pz) \leq exp(-\frac{1}{36} pz)$. For all large clusters $C^*_i$, we have $Pr(\lvert S \cap C^*_i \rvert \leq (1 - \frac{1}{4}) pn_i) \leq exp (-\frac{1}{36} p n_i) \leq exp(-\frac{1}{36}p\frac{z}{k})$.
    
    Now by union bounding over the failure events for each bin and large cluster, the probability that Condition \ref{condbin} or \ref{condlarge} does not hold is upper bounded by:
    \[\frac{n}{z} 2 \cdot exp(-\frac{1}{36} pz) + k \cdot exp(-\frac{1}{36}p\frac{z}{k})\]
    Because $p \geq \frac{36}{z} \log (\frac{4nk^2}{z})$ and $p \geq 36 \frac{k}{z} \log (2k^3)$, the first and second terms are both upper bounded by $\frac{1}{2k^2}$, respectively.
\end{proof}

For the remainder of this section, we assume that both Conditions \ref{condbin} and \ref{condlarge} hold. Now we will formalize the idea that it suffices to get a good representation of the large clusters, because we can simply throw away the remaining points as outliers by increasing the number of outliers by a constant factor.

\begin{lemma}\label{lemmacover}
	Let $X' = X_{covered}$, where $X_{covered}$ is the union of all covered large clusters (so $X'$ excludes all small clusters, all uncovered clusters, and all outliers in the optimal $z$-solution on $X$.)
	
	Then $f^X_{(9 + 7 \kappa_1)z}(C) \leq f^{X'}(C)$.
\end{lemma}
\begin{proof}
	It suffices to show that $\lvert X_{small} \cup (X \setminus X_z(C^*)) \cup X_{uncovered} \rvert \leq (9 + 7\kappa_1)z$.
	
	By definition $\lvert X_{small}\rvert \leq z$ (because each small cluster has at most $\frac{z}{k}$ points and there are at most $k$ small clusters), and $\lvert X \setminus X_z(C^*) \rvert \leq z$, so it remains to bound the size of $X_{uncovered}$.
	
	Recall that $A = S_{2.5pz}(C^*) \cap S_{2.5\kappa_1 pz}(C)$, so $\lvert A \rvert \geq s - 2.5pz - 2.5 \kappa_1 pz$.  This implies $\lvert S \setminus A \rvert \leq 2.5pz(1 + \kappa_1)$.
	
	By definition, a cluster $C_i^*$ is uncovered if $\lvert C_i^* \cap A \rvert < \frac{1}{2} \lvert C_i^* \cap S \rvert$.  This implies $\lvert C_i^* \cap S \rvert \leq 2 \lvert C_i^* \cap (S \setminus A) \rvert$
	
	By summing over all uncovered clusters, we have:
	\begin{eqnarray*} &\lvert X_{uncovered} \cap S \rvert \\
	&\leq 2 \lvert X_{uncovered} \cap (S \setminus A)\rvert \leq 2 \lvert S \setminus A \rvert \leq 5pz (1 + \kappa_1)
	\end{eqnarray*}
	
	Further, by Condition \ref{condlarge}, for every large cluster $C_i^*$, we have $\lvert S \cap C_i^* \rvert \geq 0.75 p n_i$, which gives $n_i \leq \frac{4}{3}\frac{1}{p} \lvert S \cap C_i^* \rvert$. This holds for every uncovered cluster, so we can upper bound:
	\[\lvert X_{uncovered} \rvert \leq \frac{4}{3} \frac{1}{p} \lvert X_{uncovered} \cap S \rvert \leq \frac{20}{3}z(1 + \kappa_1)\]
	Combining these bounds gives:
 \begin{eqnarray*}&&\lvert X_{small} \cup X_{uncovered} \cup (X \setminus X_z(C^*)) \rvert \\
 &\leq& z + z + \frac{20}{3}z(1 + \kappa_1)\\
 &=& (2 + \frac{20}{3}(1 + \kappa_1))z \leq (9 + 7\kappa_1)z
	\end{eqnarray*}
\end{proof}
	
Now it suffices to bound the cost of clustering all the covered clusters, which we do so with a standard moving argument:
	
\begin{lemma}\label{lemmamove}
	$f^{X'}(C) \leq 2f^X_z(C^*) + \frac{32}{3} \frac{1}{p}(f^S_{2.5 pz}(C^*) + f^S_{2.5\kappa_1 pz}(C))$
\end{lemma}
\begin{proof}
	By the approximate triangle inequality:
	\[f^{X'}(C) \leq 2f^{X'}(C^*) + 2 \sum\limits_{i \in [k]} n_i d^2(c_i^*, C),\]
	 where the first term accounts for moving each point in $X'$ to its closest center in $C^*$, and the second term accounts for moving each point from its respective center in $C^*$ to the nearest center in $C$.
	
	Note that $X' \subset X_z(C^*)$, so $f^{X'}(C^*) \leq f^{X_z(C^*)}(C^*) = f^X_z(C^*)$.
	
	It remains to bound $2\sum\limits_{i \in [k]} n_i d^2(c_i^*, C)$. By a standard averaging argument, for any covered cluster $C_i^*$:
	\[d^2(c_i^*, C) \leq \frac{1}{\lvert C_i^* \cap A \rvert} \sum\limits_{x \in C_i^* \cap A } (2d^2(x, c_i^*) + 2d^2(x, C))\]
	Because $C_i^*$ is covered, $\frac{1}{\lvert C_i^* \cap A \rvert} \leq \frac{2}{\lvert C_i^* \cap S \rvert}$. Further, by Condition \ref{condlarge}, we have $\lvert C_i^* \cap S \rvert \geq 0.75 p n_i$.
	
	Using these results, we can bound $2\sum\limits_{i \in [k]} n_i d^2(c_i^*, C) \leq \frac{32}{3} \frac{1}{p} \sum\limits_{i \in [k]} (Q_i + R_i)$, where we define $Q_i \defeq \sum\limits_{x \in C_i^* \cap A} d^2(x, c_i^*)$ and $R_i \defeq \sum\limits_{x \in C_i^* \cap A} d^2(x, C)$.
	
	Recall that $A = S_{2.5pz}(C^*) \cap S_{2.5\kappa_1 pz}(C)$, so:
	\begin{eqnarray*}
	&&\sum\limits_{i \in [k]} Q_i = \sum\limits_{x \in X' \cap A} d^2(x, C^*) \leq \sum\limits_{x \in A} d^2(x,C^*) \\
	&&\leq f^S_{2.5pz}(C^*)
	\end{eqnarray*}
	Analogously we can show $\sum\limits_{i \in [k]} R_i \leq f^S_{2.5\kappa_1 pz}(C)$. Combining these bounds gives the desired result.
\end{proof}

Note that by definition of $C$, we have $f^S_{2.5\kappa_1 pz}(C) \leq \kappa_2 Opt(S,k, 2.5 pz) \leq \kappa_2 f^S_{2.5 pz}(C^*)$.

We require one more lemma to prove Theorem \ref{samplecore}, because we must relate the $2.5 pz$-cost of clustering $S$ with $C^*$ to the $z$-cost of clustering $X$ with $C^*$. To do this we use the fact that the bins are approximately equally-represented:

\begin{lemma}\label{lemmabin}
	$f^S_{2.5 pz}(C^*) \leq 1.25 p f^X_z(C^*)$
\end{lemma}
\begin{proof}
	By Condition \ref{condbin}, no bin contributes more than $1.25 pz$ elements to $S$, so $S_{2.5 pz}(C^*)$ excludes all of $S \cap B_{b-1}$, $S \cap B_b$.
	
	Thus $S_{2.5 pz}(C^*) \subset \bigcup_{i \in [b-2]} S \cap B_i$, so we have $f^S_{2.5 pz}(C^*) \leq \sum\limits_{i \in [b-2]} f^{S \cap B_i}(C^*)$.
	
	Further, by definition of the bin division $B_1, \dots, B_b$, we have $f^X_z(C^*) = \sum\limits_{i \in [b-1]} f^{B_i}(C^*)$, and for any $i \in [b-2]$:
	\[\max\limits_{x \in S \cap B_i} d^2(x,C^*) \leq \min\limits_{x \in B_{i+1}} d^2(x,C^*)\]
	Our strategy will be to charge each point in $S \cap B_i$ to a point in $B_{i+1}$. Observe $\lvert B_i \rvert = z$ for all $i$ and by Condition \ref{condbin}, $\lvert S \cap B_i \rvert \leq 1.25 pz$. This implies $f^{S \cap B_i}(C^*) \leq 1.25 p f^{B_{i+1}}(C^*)$ for all $i \in [b-2]$.
	
	We conclude:
	\begin{eqnarray*}&& f^S_{2.5 pz}(C^*) \leq \sum\limits_{i \in [b-2]} f^{S \cap B_i}(C^*) \\
	&\leq& 1.25 p \sum\limits_{i \in [b-2]} f^{B_{i+1}}(C^*) \leq 1.25 p f^X_z(C^*)
	\end{eqnarray*}
\end{proof}

Now we are ready to put these lemmas together to prove Theorem \ref{samplecore}:

\begin{proof}[Proof of Theorem \ref{samplecore}]
	By Lemma \ref{lemmacond}, with probability at least $1 - \frac{1}{k^2}$, both Conditions \ref{condbin} and \ref{condlarge} hold. For the remainder of the proof, suppose both conditions hold.
	
	Now by chaining Lemmas \ref{lemmacover} and \ref{lemmamove}, we have:
	\begin{eqnarray*} &&f^X_{(9 + 7 \kappa_1)z}(C) \leq f^{X'}(C) \leq 2f^X_z(C^*)\\
	&&+ \frac{32}{3} \frac{1}{p}(f^S_{2.5pz}(C^*) + f^S_{2.5\kappa_1 pz}(C))
	\end{eqnarray*}
	Applying the definition of $S$, $f^S_{2.5\kappa_1 pz} (C) \leq \kappa_2 f^S_{2.5 pz}(C^*)$:
		\begin{eqnarray*}&&f^X_{(9 + 7 \kappa_1)z}(C) \leq f^{X'}(C) \leq 2f^X_z(C^*)\\
		&&+ \frac{32}{3} \frac{1}{p}(1 + \kappa_2)f^S_{2.5 pz}(C^*)	\end{eqnarray*}
	Finally, we apply Lemma \ref{lemmabin} to obtain:
	\begin{align*}
		f^X_{(9 + 7 \kappa_1)z}(C) &\leq 2f^X_z(C^*) + \frac{32}{3} \frac{1}{p}(1 + \kappa_2)(1.25 p f^X_z(C^*))\\
		&=(\frac{46}{3} + \frac{40}{3} \kappa_2) Opt(X,k,z)
	\end{align*}
	We may assume $\kappa_1, \kappa_2 \geq 1$, which completes the proof.
\end{proof}

\section{Other Experiment Results}
    \label{sec:app-exp}

\subsection{Algorithms Implemented}

We discuss each algorithm's implementation in more detail. When we ran $k$-means++, Lloyd's, $k$-means--, we terminated the execution when the objective improves less than a 1.00001 factor. 

\begin{enumerate}
    \item \algname~ (plus $k$-means++ on coreset). We added \algname~to $k$-means++ as a preprocessing step. See Algorithm~\ref{alg:nk} for its pseudo-code. Since we had to guess the value of $Opt$, we considered all possible values that are power of 2 in the range of $[n \min_{u, v \in X} d^2(u, v), n \max_{u, v \in X} d^2(u, v)]$. Occasionally, when $z$ was almost as big as $n/ k$, \algname~discarded almost all points -- such cases were considered as failure. However, if the guessed value of $Opt$ is sufficiently large, \algname~discards no points. Therefore, essentially this algorithm should be as good as running $k$-means++ on the coreset directly. 
    
    \item $k$-means++ (on the original input). The coreset is not used in this algorithm. So, we run $k$-means++ on the original input.
    
    \item $k$-means++ (on coreset). This algorithm runs $k$-means++ on the coreset. 
    \item Primal-dual of \cite{CharikarKMN01} (on coreset). The primal-dual algorithm \cite{CharikarKMN01} is executed on the coreset. This algorithm is quite involved, and therefore, we only provide the parameters we chose to run the algorithm. For the whole algorithm description, see
    \cite{CharikarKMN01}. The algorithm requires us to guess the value of $Opt$. As in the implementation of \algname, we considered all possible values that are power of 2 in the range of $[n \min_{u, v \in X} d^2(u, v), n \max_{u, v \in X} d^2(u, v)]$. This algorithm is based on a reduction to the facility location problem where one is allowed to choose as many centers as needed, but has to pay a (uniform) cost for using each center. Thus, another binary search is needed on the facility (center) opening cost.     Each outlier cost is set to $Opt / (2z)$.
    
    \item Uniform Sample (conservative uniform sampling plus $k$-means++): $k$-means++ was executed on a uniform sample consisting of points sampled with probability $1 / (2z)$.
    %    A uniform sample was obtained by sampling each point with probability $1 / (2z)$, and then,  on the sample. 
    \item $k$-means-- \cite{ChawlaG13} on coreset. This algorithm is a variant of  the Lloyd's algorithm that executes each iteration of Lloyd's excluding the farthest $z$ points. That is, the algorithm repeats the following: it bring back all input points, excludes the farthest $z$ points from the current centers, reassigns each remaining point to the closest center, and then recomputes the center of each cluster.

    \item Local search of \cite{GuptaKLMV17} (on coreset). In principle, this algorithm may end up with discarding $\Omega(z k \log n)$ points. However, it was observed that it never discarded more than $2z$ points in experimentation. We adopt the practical implementation of the algorithm described in \cite{GuptaKLMV17}. When the algorithm converges we enforce the farthest $z$ points to be the outliers. 
        %Therefore, we chose a practical implementation of the algorithm, which was running $k$-means++, discarding $z$ outliers, and then running $k$-means++ again on the remaining data. 
\end{enumerate}

\subsection{Experiment Results}

In this section we present all experiment results. 

%\subsubsection{Synthetic Data Sets}
%%%%%%%%%%%%%%%%%%%%%%%%
%running time
%%%%%%%%%%%%%%%%%%%%%%%%
%\begin{center}
\begin{table*}
%\begin{center}
\begin{tabular}{r|r|r|r|r|r|r|r}
\hline

\multirow{2}{*}{($d,k,z$)} & \algname & Primal &coreset $k$- & original & Local &  $k$- & Uniform\\ & & Dual &means++ & $k$-means++ &Search &means-- & Sample\\

% ($d,k,z$)& \algname & Primal-Dual & S+Lloyd & vanilla & Local Search & LSO & Uniform Sample  \\
 \hline 
  
(10,10,10K) & 333 & $>$ 4hrs & 223 & 110 & 347 & 357 & 14

  \\
 \hline
  
(10,10, 50K) & 89 & 5400 & 33 & 126 & 143 & $>$ 4hrs & 14

 \\
 \hline 
  
(10, 20, 10K) & 864 & $>$ 4hrs & 667 & 249 & 3712 & $>$ 4hrs & 20

  \\
\hline
 
(10, 20, 50K) & 214 & $>$ 4hrs & 102 & 718 & 3355 & $>$ 4hrs & 20
\\

\hline
 
(20, 10, 10K) & 6576 & 5759 & 306 & 141 & 500 & 519 & 24
\\
\hline
 
(20, 10, 50K) & 145 & 5855 & 58 & 180 & 269 & $>$ 4hrs & 24
\\
 \hline
(20, 20, 10K) & 1590 & $>$ 4hrs & 1232 & 270 & 6698 & 7787 & 36
\\
\hline
 
(20, 20, 50K) & 361 & $>$ 4hrs & 203 & 1173 & 5278 & $>$ 4hrs & 36
\\
\hline

\end{tabular}
%\end{center}
\caption {The running time (sec.) for synthetic data sets with noise sampled from $[-1/2, 1/2]^d$.}
\label{table:failure}
\end{table*}
%\end{center}

%\begin{center}
\begin{table*}
%\begin{center}
\begin{tabular}{r|r|r|r|r|r|r|r}
\hline

\multirow{2}{*}{($d,k,z$)} & \algname & Primal &coreset $k$- & original & Local & $k$- & Uniform\\ & & Dual &means++ & $k$-means++ &Search &means-- & Sample\\
 \hline 
  
(10,10,10K) &  285 & 11271 & 129 & 406 & 247 & 4039 & 13

  \\
 \hline
  
(10,10, 50K) &  126 & 7638 & 35 & 772 & 154 & 154 & 14
 \\
 \hline 
  
(10, 20, 10K) &  860 & $>$ 4hrs & 585 & 900 & 3966 & 3970 & 20
  \\
\hline
 
(10, 20, 50K) &  280 & $>$ 4hrs & 103 & 1742 & $>$ 4hrs & $>$ 4hrs & 20

\\

\hline
 
(20, 10, 10K) &  415 & 13973 & 153 & 557 & 328 & 5658 & 25

\\
\hline
 
(20, 10, 50K) &  220 & 11356 & 60 & 1032 & 269 & $>$ 4hrs & 25

\\
 \hline
(20, 20, 10K) &  1235 & $>$ 4hrs & 742 & 1050 & 6278 & 6629 & 36

\\
\hline
 
(20, 20, 50K) &  474 & $>$ 4hrs & 184 & 2079 & 4967 & $>$ 4hrs & 36

\\
\hline

\end{tabular}
%\end{center}
\caption {The running time (sec.) for synthetic data sets with noise sampled from $[-5/2, 5/2]^d$.}
\label{table:failure}
\end{table*}
%\end{center}

%%%%%%%%%%%%%%%%%%%%%%%%
%objective
%%%%%%%%%%%%%%%%%%%%%%%%
%\begin{center}
\begin{table*}
%\begin{center}
\begin{tabular}{r|r|r|r|r|r|r|r}
\hline

\multirow{2}{*}{($d,k,z$)} & \algname & Primal &coreset $k$- & original & Local & $k$- & Uniform\\ & & Dual &means++ & $k$-means++ &Search &means-- & Sample\\
 \hline 
  
(10,10,10K) &	1.0002&	-&	12.6865&	1.0318&	1.2852&	1.2852&	1.1501
  \\
 \hline
  
(10,10, 50K)&1.7474&	61.2273& 1.7475&	1.8293&	1.2351&	-&	144.8253

 \\
 \hline 
  
(10, 20, 10K) &	1.0002&	-&	8.3791	&1.0141	&1.2949&	-	&1.4066
  \\
\hline
 
(10, 20, 50K) &	6.9265&-&	33.7886&	29.9914&	1.2444&	-&	352.9934

\\

\hline
 
(20, 10, 10K) &	1.0002&	94.9299&	48.5067&	1.0392&	1.1547&	1.1547&	1.1499

\\
\hline
 
(20, 10, 50K) &	2.4857&	1.1048&	2.4857&	2.6371&	1.1293&	-&	76.9656

\\
 \hline
(20, 20, 10K) 	&1.0002	&-&	1.0422&	1.0369&	1.1532&1.1532	&1.3261

\\
\hline
 
(20, 20, 50K) 	&1.8725	&-	&1.8768	&42.7528&	1.1166&	-&	491.5800
\\
\hline

\end{tabular}
%\end{center}
\caption {The objective value for synthetic data sets with noise sampled from $[-1/2, 1/2]^d$.}
\label{table:failure}
\end{table*}
%\end{center}

%\begin{center}
\begin{table*}
%\begin{center}
\begin{tabular}{r|r|r|r|r|r|r|r}
\hline

\multirow{2}{*}{($d,k,z$)} & \algname & Primal &coreset $k$- & original & Local & $k$- & Uniform\\ & & Dual &means++ & $k$-means++ &Search &means-- & Sample\\
 \hline 
  
(10,10,10K) &	1.0002&	1.1138&	67.8622&	62.6428&	1.1152&	1.1152&	1.1363

  \\
 \hline
  
(10,10, 50K) &	1.0017&	82.2102&	168.1054&	101.0163&	1.1678&	1.1678&	107.6996

 \\
 \hline 
  
(10, 20, 10K) &1.0002& -&	58.4918&	74.9860&	1.1264&	1.1264&	1.3842

  \\
\hline
 
(10, 20, 50K)	&1.0018	&-&	172.8086&	188.8177&	-&	-&	348.5592
\\

\hline
 
(20, 10, 10K) &1.0002&	1.0410&	128.6235&	45.5434&	1.0432&	1.0432&	1.1440 
\\
\hline
 
(20, 10, 50K)  &1.0013&	1.0932&	159.8269&	151.0288&	1.0907&	-&	217.1559

\\
 \hline
(20, 20, 10K) & 1.0002&	-&	119.9495&	98.5209&	1.0404&	1.0404&	1.3629

\\
\hline
 
(20, 20, 50K) &	1.0016&	-&	200.1431&	242.0977&	1.0535&	-&	454.5076

\\
\hline
\end{tabular}
%\end{center}
\caption {The objective value for synthetic data sets with noise sampled from $[-5/2, 5/2]^d$.}
\label{table:failure}
\end{table*}
%\end{center}

%%%%%%%%%%%%%%%%%%%%%%%%
%precision
%%%%%%%%%%%%%%%%%%%%%%%%
%\begin{center}
\begin{table*}
%\begin{center}
\begin{tabular}{r|r|r|r|r|r|r|r}
\hline

\multirow{2}{*}{($d,k,z$)} & \algname & Primal &coreset $k$- & original & Local & $k$- & Uniform\\ & & Dual &means++ & $k$-means++ &Search &means-- & Sample\\
 \hline 
  
(10,10,10K) & 1 & - & 0.9999 & 1 & 1 & 1 & 1
  \\
 \hline
  
(10,10, 50K) & 1 & 0.7586 & 1 & 1 &
1 & - & 0.5983

 \\
 \hline 
  
(10, 20, 10K) & 1 & - & 0.9999 & 1 & 1 & - & 1
\\
\hline
 
(10, 20, 50K) & 0.9998 & - & 0.9956 & 0.9980 & 1 & - & 0.0768

\\

\hline
 
(20, 10, 10K) & 1 & 0.7070 & 1 & 1 & 1 & 1 & 1

\\
\hline
 
(20, 10, 50K) & 1 & 1 & 1 & 1 & 1 & - & 0.8830

\\
 \hline
(20, 20, 10K) & 1 & - & 1 & 1 & 1 & 1 & 1

\\
\hline
 
(20, 20, 50K) & 1 & - & 1 & 0.9999 & 1 & - & 0.2312

\\
\hline

\end{tabular}
%\end{center}
\caption {
The precision/recall value for synthetic data sets with noise sampled from $[-1/2, 1/2]^d$.
}
\label{table:failure}
\end{table*}
%\end{center}

%\begin{center}
\begin{table*}
%\begin{center}
\begin{tabular}{r|r|r|r|r|r|r|r}
\hline

\multirow{2}{*}{($d,k,z$)} & \algname & Primal &coreset $k$- & original & Local & $k$- & Uniform\\ & & Dual &means++ & $k$-means++ &Search &means-- & Sample\\
 \hline 
  
(10,10,10K) & 1 & 1 & 1 & 1 & 1 & 1 & 1

  \\
 \hline
  
(10,10, 50K) & 1 & 1 & 1 & 1 & 1 & 1 & 1

 \\
 \hline 
  
(10, 20, 10K) & 1 & - & 1 & 1 & 1 & 1 & 1

  \\
\hline
 
(10, 20, 50K) & 1 & - & 1 & 1 & - & - & 0.9999
\\

\hline
 
(20, 10, 10K) & 1 & 1 & 1 & 1 & 1 & 1 & 1

\\
\hline
 
(20, 10, 50K) & 1 & 1 & 1 & 1 & 1 & - & 1

\\
 \hline
(20, 20, 10K) & 1 & - & 1 & 1 & 1 & 1 & 1

\\
\hline
 
(20, 20, 50K) & 1 & - & 0.9999 & 1 & 1 & - & 1

\\
\hline

\end{tabular}
%\end{center}
\caption {The precision/recall value for synthetic data sets with noise sampled from $[-5/2, 5/2]^d$.}
\label{table:failure}
\end{table*}
%\end{center}

%\subsubsection{Real-world Data Sets}

%%%%%%%%%%%%%%%%%
%%%%%%%7*7%%%%%%%
%%%%%%%%%%%%%%%%%
\begin{table*}
%\begin{center}
\begin{tabular}{c|c|c|c|c|c|c|c}
\hline
\multirow{2}{*}{
                 } & \algname & $k$- & coreset $k$- & original  & Uniform  & Primal & Local \\& &means-- &  means++ &$k$-means++ & Sample &Dual &Search\\ \hline
\multirow{3}{*}{\textsc{Skin}-5} 
& 1         & 0.9740    & 1.0641    &  0.9525 & 1.1273 & 1.1934 & 0.9746\\ 
&  0.8065    & 0.7632    & 0.7653    & 0.7775 & 0.7575 & 0.7636 & 0.7632 \\
& 56        & 86        & 39        & 34 & 1 & 1274 & 86\\ \hline
\multirow{3}{*}{\textsc{Skin}-10} 
&  1 & 1.5082    & 1.4417          & 1.6676 & 1.4108 & 1.1197 & 1.5082\\
&  0.9424    & 0.9044    & 0.9012          & 0.8975 & 0.9346 & 0.9473& 0.9044\\
& 56        & 89        & 37                & 43 & 1 & 1931 & 89\\ \hline
\multirow{3}{*}{\textsc{Susy}-5} 
& 1 & 1.2096 & 1.0150 & 1.0017 & 1.1816 & 1.2093 & 1.1337 \\
&0.8518  & 0.8151 &  0.8622 & 0.8478 &  0.7622 & 0.8558 & 0.8151\\
& 1136 & 672  & 462  & 6900 & 97 & 4261 & 672\\ \hline
\multirow{3}{*}{\textsc{Susy}-10} 
&  1 & 1.1414 & 1.0091 & 1.0351 & 1.1611 & 1.2474 & 1.1414\\
&0.9774  & 0.9753 &  0.9865  & 0.9814 & 0.9816 & 0.9808 & 0.9753\\
& 1144 & 697 & 465 & 6054 & 98 & 5075 & 697 \\ \hline
\multirow{3}{*}{\textsc{Power}-5} 
&  1 & 1.0587 & 1.0815  & 1.0278 &1.2814 & 1.2655 & 1.0587 \\
& 0.6720 & 0.6857 &  0.7247 & 0.7116 & 0.7943 & 0.6481 & 0.6857\\
& 363 & 291 & 177 & 689&19 & 2494 & 291 \\ \hline
\multirow{3}{*}{\textsc{Power}-10} 
&  1 & 1.0625 & 1.0876 & 1.0535 & 1.2408 & 1.2299 & 1.0625 \\
& 0.9679 & 0.9673 &  0.9681 & 0.9649 &0.9821 & 0.9634 & 0.9673 \\
&350  &251  & 142 & 943 & 19 & 3097 & 251 \\ \hline
\multirow{3}{*}{\textsc{KddFull}} 
& 1  & 2.0259 & 1.5825 & 1.5756 & 1.1527 & 2.6394 & 2.0259\\
& 0.6187 &  0.6436 & 0.3088 & 0.3259 &0.5855 & 0.5947 & 0.6436 \\
& 1027 & 122 & 124 & 652 &104 & 844 & 122\\ \hline
\end{tabular}
%\end{center}
\caption {Experiment results on real-world data sets with $\Delta = 5, 10$. The top, middle, bottom in each entry are the objective (normalized relative to \algname), precision, and run time (sec.), resp.}
%\label{table:real}
\end{table*}

\end{document}